\documentclass[12pt,a4paper]{article}
\usepackage[centertags]{amsmath}
\usepackage{amsfonts,amsthm,amssymb}
\usepackage{mathrsfs}
\usepackage{amssymb}
\usepackage{amsmath}
\usepackage{graphicx}
\usepackage{booktabs}
\usepackage{appendix}
\usepackage{thmtools,thm-restate}
\usepackage[hmargin=3.175cm,vmargin=3.175cm]{geometry}
\usepackage{etoolbox}
\usepackage[colorlinks=true,
linkcolor=red,
urlcolor=blue,
citecolor=blue]{hyperref}
\usepackage{tikz}
\usetikzlibrary{calc, positioning, intersections, through}

\usepackage{natbib}
\newcommand{\Ex}{\mathbb{E}}
\newcommand{\Pro}{\mathbb{P}}
\newcommand{\bX}{\mathbf{X}}

\newcommand{\R}{\mathbb{R}}

\newcommand{\s}{\sigma}

\DeclareMathOperator{\supp}{supp}

\newtheorem{lemma}{Lemma}

\newtheorem{proposition}{Proposition}

\newtheorem{theorem}{Theorem}

\newtheorem{corollary}{Corollary}

\theoremstyle{definition}
\newtheorem{definition}{Definition}
\newtheorem{example}{Example}

\begin{document}
	\title{Persuading while Learning}
	
	\author{
		Itai Arieli,\thanks{Faculty of Data and Decision Sciences, Technion, E-mail: iarieli@technion.ac.il.}
		\hspace{0.1cm}
		Yakov Babichenko,\thanks{Faculty of Data and Decision Sciences, Technion, E-mail: yakovbab@technion.ac.il.}
		\hspace{0.1cm}
		Dimitry Shaiderman,\thanks{Faculty of Data and Decision Sciences, Technion, E-mail: dima.shaiderman@gmail.com.}
		\hspace{0.1cm}
		Xianwen Shi\thanks{Department of Economics, University of Toronto, E-mail: xianwen.shi@utoronto.ca. }
		\hspace{1cm}
	}
	
	\maketitle

	\begin{abstract}
		We propose a dynamic product adoption persuasion model involving an impatient partially informed sender who gradually learns the state.  In this model, the sender gathers information over time, and hence her posteriors' sequence forms a discrete-time martingale. The sender commits to a dynamic revelation policy to persuade the agent to adopt a product. We demonstrate that under the assumption that the sender's martingale possesses Blackwell-preserving kernels, the family of optimal strategies for the sender takes an interval form; namely, in every period the set of martingale realizations in which adoption occurs is an interval. Utilizing this, we prove that if the sender is sufficiently impatient, then under a random walk martingale, the optimal policy is fully transparent up to the moment of adoption; namely, the sender reveals the entire information she privately holds in every period.
	\end{abstract}
	\section{Introduction}
	In many practical applications of information design, the informed party (the sender) is not always fully aware of the realized state of the world. This information may be gradually revealed to her over time. In such dynamic persuasion, the sender simultaneously reveals information to the receiver while learning about the state.
	
	For a concrete example, consider a state's approach to COVID-19 vaccination policy amid the pandemic. On one side, the state advocated for widespread vaccination, prioritizing public inoculation despite uncertainties surrounding vaccine effectiveness. This position was justified by the established safety of vaccines, asserting that the collective benefits of maximizing vaccination coverage outweigh the potential risks of vaccine ineffectiveness. On the other side, some individuals sought assurance regarding the vaccine's efficacy before opting for vaccination. 
	Such an interaction might be captured by a binary-state environment of \emph{medium effectiveness} and \emph{high effectiveness}. While the state wants the public to be vaccinated even if the effectiveness is medial, the public agrees to be vaccinated only once their belief about the effectiveness being high exceeds a certain threshold.
	
	A fundamental question evolves around how the state should disclose information to the public regarding vaccine effectiveness, considering its objectives. Should transparency be prioritized, with all information made available, or should it carefully select the information that should be disclosed to the public?
	
	
		Other related scenarios include an entrepreneur of a new product or app who learns additional information through market testing and must decide how much of this new information to reveal to potential investors or venture capitalists. Similarly, a central bank defending a currency peg must decide how to disclose information about the foreign exchange market to potential speculators. Additionally, a central bank can choose how to reveal information about bank stress test outcomes to dissuade investors from withdrawing their funds, with the format of stress tests assumed to be fixed.
		
		In all these cases, the sender faces the challenge of acquiring information while needing to persuade the receiver during this learning phase.  The primary objective of this paper is to explore the dynamics of optimal information disclosure for an adaptive sender that steadily acquires information over time. We study this question in the context of the classic product adoption model. In particular, a sender partially learns over time information about the quality of a product that could be either of good or bad quality. The sender reveals information to a short-lived sequence of receivers. The receivers arrive sequentially, observe the past revelations, and need to decide whether to adopt the product or not.\footnote{Equivalently, our model can be viewed as an interaction with a single long-lived receiver; once the posterior of the receiver has passed the desired threshold for adoption she can stop revealing information which incentivizes the long-lived receiver to adopt immediately.}

		The quality of the product is gradually revealed to the sender through martingale realizations.
		The dilemma faced by the sender is similar to the scenarios mentioned above. A distinguished policy in this context is the \emph{greedy policy}, where the sender, at each stage, selects a policy that maximizes the probability of the receiver adopting the product in the same period. On the other extreme, the sender could choose to wait, revealing no information initially, and then utilize the accumulated information later. Indeed as we shall show both policies may be optimal. More generally, the interplay that has a significant role in our analysis is whether to reveal more today or to wait for tomorrow and learn more about the state.
		
		\paragraph{Our contribution}
		In the setting we have just described, the sender's learning process is fully specified by a belief martingale. One way to describe the martingale is by using a probability kernel, which specifies the belief distribution tomorrow based on the sender's belief distribution today and thus defines a mapping between probability measures. A notion that plays a significant role in our analysis is that of \emph{Blackwell-order preserving} kernels. A probability kernel is said to be Blackwell order-preserving if this mapping preserves the Blackwell order on probability measures. In simple terms, if we consider two senders such that the belief distribution of one of them today is Blackwell more informative than the other, 
		then it holds that the more informed sender will retain his advantage tomorrow and will remain more informed. We say that a martingale is Blackwell-order preserving if its kernels are Blackwell-order preserving for every period $t$. We show that some standard martingales, like the belief martingale generated by conditionally independent signals or the martingale which induces a random walk, are Blackwell-order preserving.
		
		To understand the role that these martingales play in our analysis, we note that in general the problem that the sender is facing can be essentially translated to an \emph{optimal stopping problem}. That is, she needs to specify the persuasion probability for every privately observed posterior. Indeed, the information revelation is public; once a receiver has adopted the product, persuasion is stopped, and all subsequent receivers will also adopt, with no further information revealed. In contrast to standard optimal stopping problems, the state variables are the set of positive measures, which makes our problem intractable in general.
		
		Our first main result, Theorem \ref{theo:main}, shows that under a Blackwell order-preserving martingale, the class of optimal policies has a simple structure which we term \emph{interval policies}. Interval policies are a class of policies in which the sender chooses to stop whenever the posterior belief lies in a certain time-dependent closed interval around the threshold probability of the receiver. This essentially transforms the stopping problem from an infinite-dimensional object to a one-dimensional object, making the problem much more tractable.
		
		Our main application of the interval stopping problem is to the case where the belief martingale is formed by a random walk on the grid. We view this case as closely related to our motivating COVID example. We further consider the case where the utility that the sender assigns to each time period is determined by a discount factor $\delta$. Our second main result, Theorem \ref{theorem:greedy} defines a threshold discount factor as a function of the grid and shows that a sender whose discount factor lies below this threshold will optimally act greedily in each period. This essentially provides a partial characterization for conditions under which full transparency of the sender is optimal.

		\subsection{Related Literature}
		
		This paper contributes to the growing literature on dynamic persuasion which extends the classical model of Bayesian persuasion (\citep{AM95, KG11}) to various dynamic settings.
		
		One branch of literature focuses on a setting in which the state of the world evolves as a Markov chain, and the sender privately observes the state in each period. \citep{RSV17} characterize cases in which the greedy policy (namely, persuading optimally in each period) is optimal for a sender with discounted utilities. This includes Markov chains with binary states and renewing Markov chains. Our problem (if we restrict attention to a sender with discounted utilities) can be viewed as a special case of this setting once we identify the Markov chain's states with the posteriors of the sender. \citep{RSV17} demonstrate that the problem becomes involved already for ternary states, while the above reduction creates a Markov chain with as many states as the number of possible posteriors of the sender along the learning process. This number is large (or even infinite) in most of our applications. 
		\citep{LS21} study, in the same setting, when is the statically optimal value achievable in the dynamic setting.
		A continuous-time analog of this setting, including some important binary-state examples has been studied by \citep{E17}.

		Several papers such as \citep{HO19,OSZ20,BRV21,M22} study models in which, similarly to ours, the sender is initially uninformed and dynamically learns the state. These models assume that the information acquired by the sender is \emph{public} to the sender and the receiver. A crucial aspect of our model is that the dynamically learned information is exogenous and \emph{private} to the sender. In particular, we assume that the learning procedure is independent 
		of the experiment that the information is revealed to the receiver.\footnote{We note that in the COVID example, one may argue that the number of vaccinated people also affects the information available to the state. While this is true, we neglect this effect in our analysis. Thus, we assume that information about the vaccine will be released (say, from other sources) regardless of the number of vaccinated individuals.}
		
		Different sender's objectives such as incentivizing exploration \citep{KMP14,CH18}, maximizing suspense and surprise \citep{EFK15}, or maximizing effort \citep{ES20} have been studied in a dynamic persuasion setting in the case where the sender is initially fully informed about the state.

		Several papers have developed models of dynamic persuasion to provide justifications for the sender-optimal equilibrium outcome with commitment 
		in the model of static persuasion. For example, Che, Kim and Mierendorff (2023) show in a model of dynamic persuasion where information generation and processing are both costly for the sender and the receiver that the sender-optimal outcomes in Kamenica and Gentzkow (2011) can be approximated in a Markov perfect equilibrium as the persuasion cost vanishes.\footnote{See Honryo (2018) for an earlier contribution to dynamic persuasion with persuasion cost, and Escudé and  Sinander (2023) for a model of slow persuasion where the sender is restricted to a graduality constraint.} Best and Quigley (2023) and Mathevet, Pearce, and Stacchetti (2022) show that, in a model where a long-lived sender plays a cheap-talk game against a sequence of short-lived receivers, the sender-optimal equilibrium can be supported by reputation. 
		
		Finally, besides the persuasion literature, there is also the literature on the disclosure of verifiable information in dynamic settings. This includes \citep{A15,K24}. Notice that the disclosure problem restricts the sender's strategies to the timing of revealing any evidence while the persuasion problem allows for much richer policies for the sender which include a partial revelation of these pieces of information and their timing.
		
		\section{Model and Preliminaries}\label{sec:preliminaries}
		
		
		In our setting, the sender (she) dynamically receives information about an unknown state $\omega\in \{0,1\}$. 
		Beliefs about the state $\omega$ are identified with $[0,1]$, where $x\in [0,1]$ captures the probability that is assigned to state $\omega=1$.
		Specifically, we assume that the process according to which the sender learns about  the state is given by a (commonly known) martingale $\bX=(X_t)_{t=0,1,...,T}$ supported on the interval $[0,1]$ with $X_0$ being the Dirac measure on the common prior $\pi$. The number of steps in the learning process can be either finite ($T\in  \mathbb{N}$) or infinite ($T=\infty)$.
		We restrict attention to cases in which $\bX$ is a \emph{Markovian} martingale. That is, the behavior of the martingale from time $t+1$ on depends only on the realization of $X_t$ (rather than the realization of the entire history $(X_1,...,X_t)$). 
		
		In this paper, the sender interacts with a sequence of short-lived receivers who have a binary action set $A=\{0,1\}$. In each time $t=1,2,...$ the corresponding receiver chooses an action $a_t\in \{0,1\}$. 
		Action $a_t=1$ will be called \emph{adoption}.
		The instantaneous receiver's utility is 
		\begin{align*}
			u(\omega,a_t)=u_t(\omega,a_t)= \begin{cases}
				0 &\text{ if } a_t=0 \\
				1 &\text{ if } a_t=\omega=1 \\
				-\frac{l}{1-l} &\text{ if } a_t=1, \omega=0.
			\end{cases} 
		\end{align*}
		
		Namely, at any time $t$, the corresponding receiver prefers adoption if and only if his belief lies in $[l,1]$.

			
			
			The sender wishes to exploit her learning about the state, to encourage the receiver to adopt. To do so, she has to decide on the amount of information she wishes to reveal about her learning, which in turn generates the receiver's belief process about the state. A baseline observation, stemming from the martingale property of $\bX$, is that the sender may keep the receiver's belief fixated by refraining from revealing information. In the given model, this observation implies that upon adoption at some stage $t$, the sender can guarantee adoption in all subsequent stage as well.
			
			Therefore, the only relevant sequences of actions for which we shall define the sender's utility are of the form $(a_t)_{t=1,2,\ldots,T}=(0,...,0,1,1,1,...)\in\{0,1\}^T$.
			We denote by $w_t$ the sender's utility from adoption at time $t$, where $w_t$ is a decreasing sequence.  Alternatively, we may think of $w_t-w_{t+1}$ as the mass of consumers that arrive at time $t$ (where $w_{T+1}=0)$, and thus the utility of the sender equals the expectation of $\sum_{t=1}^T (w_t-w_{t+1}) a_t$.
			
			One particular example is $w_t=\delta^{t}$ which corresponds to the $\delta$-discounted utilities for the sender.

			\subsection{Persuasion as a stopping problem} 
			The sender commits to an information revelation policy which is a mapping from realizations of the martingale $(x_1,...,x_t)$ to distributions over abstract signals. By applying standard direct revelation arguments, the abstract signals can be replaced by the recommendation for adoption or waiting. In this equivalent formulation, an information revelation policy is simply a stopping time
			$\tau$ on the martingale $\bX$. To understand this equivalence note that once the receiver has adopted the product at time $t$ it is optimal for the sender to reveal no further information at subsequent periods.
			
			Formally, we identify the sender's policy with a \emph{randomized} stopping rule defined in terms of a random time $\tau$ and measurable mappings $\tau_t:[0,1]^t\to[0,1]$, $t=1,\ldots,T$ so that  $\mathbb{P}(\tau=t|\tau\geq t,X_1=x_1,\ldots,X_t=x_t) = \tau_t(x_1,\ldots,x_t)$. That is, conditional on not stopping prior to time $t$ the stopping probability is determined by $\tau_t$. Moreover, we require our stopping rule to satisfy an \emph{incentive compatibility constraint} in the form of:
			\begin{align}\label{Incentive Compatibility Constraint}
				\mathbb{E} [X_{t}|\tau=t]\geq l, \quad \forall t=1,...,T.
			\end{align}  
			This constraint asserts that conditional on stopping (recommendation to adopt) at time $t$, the receiver will be (weakly) better off following this recommendation. This is because, by Bayes rule, the posterior of the receiver given recommendation at time $t$, we have $\Pro(\omega=1|\tau=t) = \mathbb{E}[X_{t}|\tau=t]$. Let us denote by $\mathcal{T}$ the set of all randomized stopping rules $\tau$ satisfying \eqref{Incentive Compatibility Constraint}.

			We note that in any optimal solution, the incentive compatibility constraint will always be satisfied as equality since, otherwise, the sender gives information ``for free.'' 
			In addition, our assumption that $\pi<l$ implies that $\mathbb{E}[X_t|\tau>t]<l$ and thus the receiver  will take action $1$ by time $t$ iff $\tau\leq t$. 
			The sender's optimization problem is thus\footnote{For the case $T<\infty$, the supremum can be replaced by a maximum using backward-induction arguments. For the case $T=\infty$ we prove the existence of the maximum for all the special cases that we consider.} 
			\begin{align}\label{eq:sender-u}
				V^* = \sup_{\tau \in \mathcal{T}} \left\{\sum_{t=1}^T \Pro[\tau = t] w_t \right\}.
			\end{align}
			
			\begin{example}\label{ex1}
				To demonstrate the subtleties in the above persuasion problem, consider a two-period model where the sender who is sequentially exposed to two conditionally independent signals $S_1$ and $S_2$ with binary support $\{L,H\}$, one signal each period. The state is equally likely ex ante, i.e., $\pi=1/2$. For $t=1,2$, the distribution of signal $S_t$  conditional on state $\omega$ is  $\Pro(S_t=H|\omega=1) = \Pro(S_t=L|\omega=0)=q_t$  with $q_t\in [1/2,1]$. Let 
				$$x_s = \Pro(\omega=1|S_1=s), s\in \{L,H\}$$ 
				be the sender's period-one posterior estimate of the state after receiving signal $S_1=s$ and 
				$$x_{ss'} = \Pro(\omega=1|S_1=s, S_2=s'), s,s'\in \{L,H\}$$ 
				be the period-two posterior estimate after receiving signals $S_1=s$ and $S_2=s'$.  Then the martingale $\bX=(X_0,X_1,X_2)$ is given by $\supp(X_0)=\{\pi\}, \supp(X_1) = \{x_L,x_H\}$, and $\supp(X_2) = \{x_{LL},x_{LH},x_{HL}, x_{HH} \}$. We normalize the sender's utility by setting $w_1=1$. 
				
				Since there are only two periods, the optimal policy must be greedy at $t=2$. That is, in the second period the sender persuades the receiver with the maximal possible probability.

				Is greedy policy also optimal at $t=1$? Suppose that $l=18/25$, the information at the first period is generated by a signal with precision $q_1=3/4$ and in the second period by a conditionally independent signal with precision $q_2=4/5$. Then for $w_2\in [0,0.618]$,\footnote{The numerical cutoff is approximate.} the optimal policy is greedy in $t=1$ with $|\nu_1|=25/47$. If $w_2\in [0.955, 1]$, the optimal policy stays mute at time $t=1$ (i.e.,  $|\nu_1|=0$). If $w_2\in (0.618,0.955)$, the optimal policy persuades at time $t=1$, but not greedily with $|\nu_1|=25/58<25/47$.
				\begin{center}
					
					\begin{tikzpicture}
						[dot/.style={circle,inner sep=1pt,fill,label={#1},name=#1},
						help lines/.style={blue!10,ultra thin},scale=3,
						extended line/.style={shorten >=1cm, shorten <=-2.2cm}, 
						extended line/.default=1cm]
						
						\draw[-, black, dashed] (0,2) -- (2,2) node [black, right]{\small $t=0$};
						\draw[-, black, dashed] (0,1) -- (2,1) node [black, right]{\small $t=1$};
						\draw[-, black, dashed] (0,0) -- (2,0) node [black, right]{\small $t=2$};
						
						\draw [-, black, thin] (1,2) -- (0.6,1) node [blue!80!black, xshift=0.5cm, yshift=0.3cm]{\small $x_L$};
						\draw [-, black, thin] (1,2) -- (1.4,1) node [blue!80!black, xshift=0.5cm, yshift=0.3cm]{\small $x_H$};
						\draw [-, black, thin] (0.6,1) -- (0.1,0) node [blue!80!black, xshift=0cm, yshift=-0.3cm]{\small $x_{LL}$};
						\draw [-, black, thin] (0.6,1) -- (1.4,0) node [blue!80!black, xshift=0.1cm, yshift=-0.3cm]{\small $x_{HL}$};
						\draw [-, black, thin] (1.4,1) -- (0.8,0) node [blue!80!black, xshift=0cm, yshift=-0.3cm]{\small $x_{LH}$};
						\draw [-, black, thin] (1.4,1) -- (1.9,0) node [blue!80!black, xshift=0cm, yshift=-0.3cm]{\small $x_{HH}$};
						\draw [-, blue, dashed] (1.2,2) -- (1.2,0) node [blue!80!black, xshift=0cm, yshift=-0.3cm]{\small $l$};
						
						\filldraw[red](1,2) circle (0.5pt) node [blue!80!black, xshift=0cm, yshift=0.2cm]{\small $\pi$};
						\filldraw[red](0.6,1) circle (0.5pt);
						\filldraw[red](1.4,1) circle (0.5pt);
						\filldraw[red](0.1,0) circle (0.5pt);
						\filldraw[red](0.8,0) circle (0.5pt);
						\filldraw[red](1.4,0) circle (0.5pt);
						\filldraw[red](1.9,0) circle (0.5pt);
					\end{tikzpicture}
					
				\end{center}


				Another interesting phenomenon in this simple example is the non-monotonic relationship between the optimality of the greedy policy and the precision of the second signal. If we set $q_2 = \frac{1}{2}$ or $q_2 = 1$, the greedy policy is optimal for all $w_2$ in both cases. However, as shown above, this is not true for $q_2 = \frac{4}{5}$.
				
				Solving these optimization problems can be done by observing that the only free parameter of the optimization problem is the probability of stopping at time $t=1$. Once this parameter is set, the conditional expectation of stopping being $l$ uniquely defines the probability of stopping at $x_L$ and at $x_h$. At time $t=2$, the policy is greedy with respect to the remaining mass.
				
			\end{example}
			
			

			\subsection{Measure Theoretic Formulation}
			
			A probability kernel is a measurable function $\sigma:[0,1]\to\Delta([0,1])$ such that 
			$\mathbb{E}_{X\sim \sigma(x)}[X]=x$ for every $x\in[0,1]$. That is, the expectation under the probability measure $\sigma(x)$ is $x$ for any $x\in[0,1]$. A probability kernel defines an operator from $\Delta([0,1])$ to itself where for every probability measure $\mu$, $\sigma\circ \mu\in\Delta([0,1])$ is the \emph{pushforward} probability measure of $\mu$ by $\sigma$ such that for every Borel measurable set $B\subseteq [0,1]$ it holds that $$\sigma\circ\mu\, (B)=\int_{[0,1]}\sigma(x)(B)\mathrm{d}\mu(x).$$
			
			We note that a probability measure $\mu_1\in\Delta([0,1])$ together with $T$ probability kernels $\sigma_1,\ldots,\sigma_{T-1}$ determines a martingale $\bX=(X_t)_{t=0,1,...,T}$ where 
			$X_1\sim\mu_1$ and for every $1\leq t< T$
			$$\Pro_{X_{t+1}}(\ \cdot \ |X_{t}=x_{t},\ldots,X_1=x_1)=\Pro_{X_{t+1}}(\ \cdot \ |X_{t}=x_{t})=\sigma_t(x_t).$$
			That is, the conditional distribution of $X_{t+1}$ given $X_{t}=x_{t},\ldots,X_1=x_1$ is determined by $\sigma_t(x_t).$ We note that for every $t=2,\ldots ,T-1$ the distribution $\mu_t$ of $X_t$ is given by $\sigma_{t-1}\circ\ldots\circ\sigma_{1}\circ\mu_1.$

			Conversely, for a given Markovian martingale $\bX=(X_t)_{t=0,1,...,T}$ supported on $[0,1]$, there exist probability kernels $\sigma_1,\ldots,\sigma_{T-1}$ such that if $X_1\sim\mu_1$ then the martingale generated by the above procedure is $\bX=(X_t)_{t=0,1,...,T}$. This can be easily shown by taking $\sigma_t(x_t)=\Pro_{X_{t+1}}(\ \cdot \ |X_t=x_t)$. 
			Note that for every $1\leq t\leq T-1$ the kernel $\sigma_t$ is a.s. uniquely defined on the support of $X_t$ and can be arbitrarily defined outside of the support.
			
			For any positive and finite measure $\nu$ over $[0,1]$ we let $|\nu|=\nu([0,1])$ and $\overline{\nu}$ be the expectation of its normalization $\frac{1}{|\nu|}\nu$.
			We next represent an alternative formulation to the optimization problem. Given two positive measures $\nu, \mu$ on $[0,1]$ we denote $\nu\leq\mu$ if $\nu(B)\leq\mu(B)$ for any Borel measurable set $B\subseteq[0,1]$.
			\begin{restatable}{lemma}{firstlem}
				\label{lem:altern}
				An equivalent reformulation of the sender optimization problem given in Equation  \eqref{eq:sender-u} is the following: maximize $\sum_{t=1}^T |\nu_t|w_t$ subject to the following recursively defined constraints: $X_1\sim\mu_1$,  $\nu_t\leq\mu_t$ for every $t=1,\ldots ,T$, $\overline{\nu_t}\geq l$ for every $t=1,\ldots ,T$, and $\mu_t :=\sigma_{t-1} \circ (\mu_{t-1} - \nu_{t-1})$ for every $t=2,\ldots,T$.
			\end{restatable}

			In this formulation, we have a mass $\mu_t$ that evolves according to the probability kernels $\sigma_t$, and in each step, we decide which parts of this mass, as described by $\nu_t$, to eliminate from the process to induce an immediate adoption. The proof of the lemma, as well as the proofs of all other lemmas, are relegated to Appendix \ref{ap:sec2}.


			\subsection{Blackwell Order Preserving Kernels}
			A concept that will play a fundamental role in our analysis is Blackwell order preserving probability kernels. 
			
			Consider two probability measures $\mu,\nu\in\Delta([0,1])$ with the same mean
			$\overline{\mu}=\overline{\nu}$. Recall that $\nu$ dominates $\mu$ with respect to the Blackwell ordering, denoted as $\mu \preceq_B \nu$, if $\nu$ is a \emph{mean preserving spread} of $\mu$, or equivalently, if there exists a probability kernel $\sigma$ such that $\sigma\circ\mu=\nu.$ In terms of random variables one has $\mu \preceq_B \nu$ if and only if there exist random variables $X$ and $Y$, where $X \sim \mu$, $Y \sim \nu$ and $\Ex [Y\,|\, X] = X$. In functional form, the order $\mu \preceq_B \nu$ is equivalent to requiring that $\nu (f) \leq \mu(f)$ for every non-decreasing concave $f:\mathbb{R} \to \mathbb{R}$, where we denote $\rho(f) = \int_{\mathbb{R}} f(x) \mathrm{d}\rho (x)$ for a positive measure $\rho$ defined on $\mathbb{R}$.
			
			We extend the Blackwell order to (finite) positive measures on $[0,1]$. For two positive measures $\mu,\nu$ we write $\mu \preceq_B \nu$ if $|\mu|= |\nu|$ and $\frac{\mu}{|\mu|} \preceq_B \frac{\nu}{|\nu|}$. 
			
			We next present the central notion that we shall use throughout the paper.
			\begin{definition}
				A probability kernel $\sigma:[0,1]\to\Delta([0,1])$ is called \emph{Blackwell order preserving} kernel if for every two probability measures $\mu,\nu\in\Delta([0,1])$ such that $\overline{\mu}=\overline{\nu}$ and $\mu \preceq_B \nu$ it holds that $\sigma\circ\mu \preceq_B \sigma\circ\nu$.  A Markovian martingale is called \emph{Blackwell order preserving} if $\sigma_t$ is\footnote{In case where $\sigma_t$ is not uniquely defined we only require that there exists a version of $\sigma_t$ that is Blackwell preserving.} Blackwell preserving for every time $t=1,\ldots T-1.$  
			\end{definition}
			There is a simple interpretation of Blackwell order-preserving kernels. Consider $\mu,\nu\in\Delta([0,1])$ such that $\overline{\mu}=\overline{\nu}$. As above, the measures can be identified with a signal that generates the distribution of the posterior probability to the event $\{\omega=1
			\}$. Consider a decision maker with a certain utility $u:\Omega\times A\to\mathbb{R}$ that observes a signal and then takes an expected utility-maximizing action. The signal induced by $\nu$ is better for the decision maker than the signal induced by $\mu$ if their expected utility from observing the signal induced by $\nu$ is higher than their expected utility from observing the signal induced by $\mu$.
			
			According to the Blackwell order characterization, \citep{blackwell1953equivalent} it holds that $\mu \preceq_B \nu$ if and only if the signal induced by $\nu$ is better for \emph{any} decision maker than the signal induced by $\mu$.
			
			Assume now that the decision-maker receives information in two periods where the second period's information is generated by the kernel $\sigma$ based on the first period. The Blackwell order-preserving condition over $\sigma$ simply requires that if all decision-makers prefer certain information in the first period, this preference will be preserved in the second period. We note that a probability kernel preserves the Blackwell ordering between probability measures if and only if it preserves it between positive measures with equal mass.
			
			The following lemma shows that the Blackwell preserving property should be verified for binary-supported measures only. That is, the lemma says that in order to verify that a kernel is Blackwell order preserving it is (necessary and) sufficient to show that the push forward of a binary-supported distribution Blackwell dominates the push forward of its expectation. 
			
			\begin{restatable}{lemma}{seclem}\label{lemma:blacwell preservingK}
				A probability kernel $\sigma$ is Blackwell preserving if and only if for every binary supported $\mu\in\Delta([0,1])$ it holds that $$\sigma\circ\delta_{\overline{\mu}}\preceq \sigma\circ\mu.$$
				Similarly a martingale $(X_t)_{t=1,\ldots,T}$ is Blackwell preserving if and only if for every $t\leq T-1$ the kernel $\sigma_t$ satisfies the above condition on every binary supported measure $\mu$ such that both the support of $\mu$ and $\overline{\mu}$ is contained in the support of $X_t$. 
				
			\end{restatable}
			\subsection{Examples of Blackwell preserving martingales}\label{sec:BP-martingales}
			\paragraph{Conditionally Independent Signals}
			As a first example of a Blackwell preserving martingale consider a kernel that is obtained by receiving a conditionally independent signal as a function of the prior on $\{0,1\}$. More precisely, let $S$ be some measurable signal space and let $G:\{0,1\}\to\Delta(S)$ be a probability kernel. 
			We note that a prior $y\in[0,1]$ over the set $\Omega=\{0,1\}$ together with $G$ generate a probability distribution $\Pro_y\in\Delta(\Omega\times S)$. Let $p_y(s)=\Pro_y(\omega=1|s)$ be the conditional probability of $\{\omega=1\}$ given $s$. Denote $\sigma_y$ the posterior distribution of $p_y(s)$. That is, for every Borel subset $A\subseteq[0,1]$, $\sigma_y(A)=\Pro_y(p_y(s)\in A)$. Consider the probability kernel $\s: [0,1] \to \Delta(\{0,1\})$ defined by $\s: y \mapsto \s_y$. By the law of iterated expectation, this is indeed a probability kernel. 
			\begin{restatable}{lemma}{thlem}\label{lemma:conditional ind} Consider a martingale $\bX=(X_t)_{t=1,\ldots,T}$ that is generated by the kernels $(\sigma_t)_{t=1,\ldots,T-1}$ such that for every $t=1,\ldots,T-1$ the kernel $\sigma_t$ represent conditionally independent signal, then the martingale is a Blackwell order preserving martingale.   
			\end{restatable}

			
			\paragraph{Random Walk on a Grid}
			Consider  a discrete set $\Gamma=\{z_i\}_{i\in Z}\subseteq[0,1]$ such that $Z=\mathbb{Z}\cap [a,b]$ (where $a$ and $b$ might be finite or equal $-\infty$ and $+ \infty$ respectively) $z_i>z_j$ for every $i>j \in Z$. Let $\sigma$ be a kernel that represents a random walk on $\Gamma$. That is, $\sigma(z_i)=\delta_{z_i}$ if $i\in \{a,b\}$ and $\sigma(z_i)=\frac{z_{i+1}-z_i}{z_{i+1}-z_{i-1}}\delta_{z_{i-1}}+\frac{z_i-z_{i-1}}{z_{i+1}-z_{i-1}}\delta_{z_{i+1}},$ otherwise. 
			A martingale $\bX=(X_t)_{t=0,1,...,T}$ is a \emph{random walk} if $X_0=z_i$ with probability $1$ for some $z_i\in \Gamma$ and $\sigma_t=\sigma$ for every $t=1,\ldots,T-1.$ 
			Random walk on the grid plays a fundamental role in our analysis. It corresponds to the discrete version of the Brownian motion kernels considered in \citep{HO19,OSZ20,BRV21,M22}.

			\begin{restatable}{lemma}{folem}
				A martingale that is induced by a random walk on a grid is Blackwell order-preserving.    
			\end{restatable}

			
			\subsection{Classes of policies} 
			In this section we provide two classes of policies which will be discussed throughout the paper.

			\subsubsection{The greedy policy}\label{sec:greedy}
				
				Let $\tau$ be the stopping rule defined in terms of the function $\tau_1$ that satisfies $\tau_1(x_1)=1$ if
				$x_1\in (y,1]$, $\tau_1(x_1)=0$ if $x_1\in [0,y)$, where $y$ is chosen to satisfy $\Ex[X_1|\tau=1]=l$.
				
				We note that the point $y$ and the probability $\tau_1(y)$ are almost surely uniquely determined. 
				
				The above stopping rule has the property that it maximizes the probability  $\Pro(\tau=1)$ across all stopping rules $\tau$ for which $\Ex [X_1|\tau=1]\geq l$. In particular, it solves the optimization problem for the weights $w_1=1$, $w_t = 0$, $t \geq 2$. For this reason, we refer to the function $\tau_1$ as the \emph{greedy policy} with respect to $\mu_1$.
				
				In the equivalent reformulation, the measure $\nu_1 = \Pro [\tau=1] \, \Pro[\ \cdot \ | \tau=1]$, which is eliminated from $\mu_1$ by the above stopping rule is referred to as the \emph{greedy measure}. A well-known fact is that the greedy measure is obtained by taking the top $q$-quantile of the probability measure $\mu_1$ for which the points above the quantile have a conditional expectation $l$. More precisely, let $F$ be the CDF of $\mu_1$ and let $F^{-1}(x)=\inf\{y|F(y)\geq x\}$. Then, for a $[0,1]$-uniform random variable $U$ it holds that the random variable $X=F^{-1}(U)\sim \mu_1$. The corresponding 
				$q$-quantile is the unique value $p$ such that $\Ex [X|U\geq p]=l$. Using such an approach one may compute $y$ and $\tau_1(y)$ that define the greedy policy in terms of $p$ and $F$. Namely, if $\mu_1$ has no atoms, one defines $y := F^{-1}(p)$ and $\tau_1(y)$ arbitrarily. In case $F^{-1}(p)$ is an atom of $\mu_1$, one defines $y := F^{-1}(p)$, and the probability $\tau_1 (y)$ of stopping on $y$ may be verified to be given by the expression:
				\begin{align*}
					\tau_1 (y) = \frac{F(F^{-1}(p)) - p}{\mu_1(\{F^{-1}(p)\})}.
				\end{align*}

				\begin{definition}
					For a general martingale $\bX=(X_t)_{t=1,\ldots,T}$ the \emph{greedy policy} $\tau$ is defined recursively where $\tau_1$ is the greedy policy with respect to $X_1\sim\mu_1$ and for every $t\geq 2$, $\tau_t$ is taken to be the greedy policy with respect to the measure $\mu_t=\Pro_{X_t}[\ \cdot \ |\tau\geq t]\Pro[\tau\geq t]$.
				\end{definition}
				
				\subsubsection{Interval policies}
				We next define a class of policies that contains the greedy policy as a special case.
				\begin{definition}
					Let $\bX=(X_t)_{t=1,\ldots,T}$ be a martingale. A stopping rule $\tau$ is called an  \emph{interval stopping rule} if there exists a sequence of intervals $\{[\underline{y}_t,\overline{y}_t]\}_{t=1,\ldots,T}$ such that for every $t=1,\ldots,T$ it holds that 
					\begin{itemize}
						\item[(i)] $\tau_t(x_1,\ldots,x_t)=1$ if $x_t\in(\underline{y}_t,\overline{y}_t)$,
						\item[(ii)] $\tau_t(x_1,\ldots,x_t)=0$ if $x_t\notin [\underline{y}_t,\overline{y}_t]$,
						\item[(iii)] $\mathbb{E}[X_t|\tau=t]=l$.
					\end{itemize}
				\end{definition}
				
				In terms of the mass elimination reformulation,  choosing the eliminated sequence of masses $\{\nu_t\}_{t=1,\ldots, T}$ such that $\nu_t((\underline{y}_t,\overline{y}_t))=\mu_t((\underline{y}_t,\overline{y}_t))$,
				$\nu_t([0,1]\setminus[\underline{y}_t,\overline{y}_t])=0$, and $\overline{\nu}_t=l$ for some sequence of intervals $\{[\underline{y}_t,\overline{y}_t]\}_{t=1,\ldots,T}$ induces an almost sure unique interval stopping rule $\tau$.

				Essentially, an interval stopping rule $\tau$ can be characterized by the interval
				$[\underline{y}_t,\overline{y}_t]$ at which it stops at time $t$ and the probabilities $\tau_t(\{\underline{y}_t\})=\underline{\gamma}$, $\tau_t(\{\overline{y}_t\})=\overline{\gamma}$ that play a role only if $\mu_t$ has atoms on either $\underline{y}_t \text{ or }\overline{y}_t$. Generalizing the top quantile approach discussed for the greedy measure, we note that at each time $t$ an interval stopping rule may be identified with the two top quantiles $\{\underline{q}_t,\overline{q}_t\}$ of $\mu_t$ it induces. Formally, if we let $F_{\mu_t}$ be the CDF of $\mu_t$ then $\underline{y}_t = F^{-1}_{\mu_t} (\underline{p}_t) $ and $\overline{y}_t = F^{-1}_{\mu_t} (\overline{p}_t)$, where $\Ex_{X \sim \mu_t} [X|U \in (\underline{p}_t, \overline{p}_t)]=l$, $\underline{q}_t = \Ex_{X \sim \mu_t} [X|U \geq \underline{p}_t]$, $\overline{q}_t = \Ex_{X \sim \mu_t} [X|U \geq \overline{p}_t]$, 
				$$\frac{F_{\mu_t}(F^{-1}_{\mu_t}(\underline{p}_t)) -\underline{p}_t }{\mu_t(\{F^{-1}_{\mu_t}(\underline{p}_t) \})} = \underline{\gamma},$$ and $$\frac{F_{\mu_t}(F^{-1}_{\mu_t}(\overline{p}_t)) -\overline{p}_t }{\mu_t(\{F^{-1}_{\mu_t}(\overline{p}_t) \})} = \overline{\gamma}.$$

				
				\subsubsection{Properties of interval policies}\label{Properties of interval policies}
				
				\paragraph{Existence and uniqueness}
				In the next lemma, we show that an interval stopping is uniquely determined by the mass (i.e., the probability) of stopping. 
				
				\begin{restatable}{lemma}{lemexist}\label{lem:exist}
					We let $\alpha$ be the persuasion probability of the greedy policy with respect to $\mu_1$. For any $0<\beta\leq\alpha$ there exists an almost surely unique interval stopping rule $\tau$ for which $\Pro(\tau=1)=\beta$.      
				\end{restatable}

				\paragraph{Minimality}
				We recall that the Blackwell ordering $\preceq_B$ is extendable to positive finite measures as follows. For any two positive measures $\tilde\nu,\nu'$ on $[0,1]$ with $\overline{\tilde \nu} = \overline{\nu'}$,  we write $\tilde\nu \preceq_B \nu'$ if $  \frac{1}{|\tilde\nu|}\tilde\nu \preceq_B \frac{1}{|\nu'|}\nu'$. That is, $\tilde\nu \preceq_B \nu'$ if the probability measure $\frac{1}{|\nu'|}\nu'$ forms a mean-preserving spread of the probability measure $\frac{1}{|\tilde\nu|}\tilde\nu$. 
				
				The Blackwell ordering relation forms a partial order relation in the class of all stopping rules $\nu \leq \mu_1$ of a given mass. In the following lemma we establish that the interval stopping rules are the minimal elements in this partial order. 
				
				
				\begin{restatable}{lemma}{lemods}
					\label{lem:obs}
					Let $\tilde\nu\leq \mu_1$ be the measure that corresponds to the interval stopping rule that stops on $[\underline y, \overline y]$ with $\overline{\tilde \nu}=l$. If $\nu'\leq \mu_1$, $\overline{\nu'}=l$,  is another stopping rule with $|\nu'|=|\tilde\nu|$ then $\tilde\nu\preceq_B \nu'$.
				\end{restatable}

				\section{The Optimality of Interval Policies}\label{sec:inter}

				Our first main result demonstrates a strong connection between interval-stopping rules and Blackwell order-preserving martingales. 
				\begin{theorem}\label{theo:main}
					If $\bX=(X_t)_{t=1,\ldots,T}$ is Blackwell order preserving martingale, then there exists an optimal interval stopping rule for the sender. 
				\end{theorem}
				This informational finding has an important practical implication. It  reduces the optimization problem stated in Lemma \ref{lem:altern} to finding the optimal sequence of stopping masses $\{|\nu_j||\,:\, j \geq 1\}$, rather then the optimal sequence of measures $\{\nu_j \,:\, j \geq 1\}$. Thus when an optimal interval policy exists at any time period the sender only needs to decide at every time period $t$ what mass of the available persuasion mass should she used to persuade the receiver today. Once the decision has been made the interval measure $\nu_j$ is defined uniquely.

				Already in simple scenarios such as Example \ref{ex1} (see also Examples \ref{ex2} below) we saw that the stopping mass in each period might be quite a complicated object and it might depend on delicate details of the instance. Our first main Theorem is applicable in the quite general class of Blackwell-preserving martingales. It does not characterize the stopping masses. Instead, it uniquely specifies the optimal policy as a function of these stopping masses.

				Note that the theorem also implies that an optimal interval policy exists for the case $T=\infty$. Thus the $\sup$ in Equation \eqref{eq:sender-u} can be replaced with $\max$ for Blackwell order-preserving kernels.
				
				
				\noindent \textbf{The informational interpretation.}
				Theorem \ref{theo:main} identifies an informational measurement with respect to which the sender's use of information should be minimal. Indeed, as discussed in Subsection \ref{Properties of interval policies}, the interval policies possess the property of minimality among all stopping rules, with respect to the Blackwell ordering partial order relation.

				

				\begin{example}\label{ex2}
					
					We demonstrate here how to utilize Theorem \ref{theo:main} in an example with $T=2$. We set $w_1=1$. Suppose that the sender receives two conditionally independent signals, $S_1 \in [0,1]$ in period 1 and $S_2 \in \{0,1, \phi\}$ in period 2. The common prior over the binary state is $\pi=1/2$. The distribution of the first period signal $S_1$, conditional on the state, is given by 
					\[\Pro(S_1 \le s|\omega =1) = s^2\quad \text{and} \quad \Pro(S_1 \le s|\omega =0) = s(2-s).\]
					It is easy to verify that signal $S_1$ induces a uniform posterior over $[0,1]$. The second period signal $S_{2}$ is either a perfect one or a white noise: with probability $q\in [0,1]$, signal $S_2$ perfectly reveals the state $\omega \in \{0,1\} $, and with probability $1-q$, $S_2$ is a null signal (denoted by $\phi$) which does not contain any information. 
					
					It is easy to verify that the martingale $\bX=(X_1,X_2)$ must satisfy $X_1(S_1) = S_1$ and
					\begin{equation*}
						X_{2}(S_1,S_2) =
						\left\{ 
						\begin{array}{ccc}
							1 & \text{if} & S_2=1 \\ 
							0 & \text{if} & S_2=0 \\ 
							S_{1} & \text{if} & S_2=\phi%
						\end{array}%
						\right.
					\end{equation*}
					Therefore, $X_1$ is uniformly distributed on $[0,1]$, and $X_2$ is also uniformly distributed but with atoms at $0$ and $1$ splitting a total mass of $q$. 
					
					\begin{center}
						
						\begin{tikzpicture}
							[dot/.style={circle,inner sep=1pt,fill,label={#1},name=#1},
							help lines/.style={blue!10,ultra thin},scale=4,
							extended line/.style={shorten >=1cm, shorten <=-2.2cm}, 
							extended line/.default=1cm]
							
							\draw[-, black, dashed] (0,2) -- (2,2) node [black, right]{\small $t=0$};
							\draw[-, black, dashed] (0,1) -- (2,1) node [black, right]{\small $t=1$};
							\draw[-, black, dashed] (0,0) -- (2,0) node [black, xshift=0.6cm, yshift=0.3cm]{\small $t=2$};
							
							\draw [-, black, dashed] (1,2) -- (0,1) node [blue!80!black, xshift=0.5cm, yshift=0.3cm]{};
							\draw [-, black, dashed] (1,2) -- (0.5,1) node [blue!80!black, xshift=-0.2cm, yshift=0.2cm]{};
							\draw [-, black, dashed] (1,2) -- (1,1) node [blue!80!black, xshift=0.5cm, yshift=0.3cm]{};
							\draw [-, black, dashed] (1,2) -- (1.5,1) node [blue!80!black, xshift=0.5cm, yshift=0.3cm]{};
							\draw [-, black, dashed] (1,2) -- (2,1) node [blue!80!black, xshift=0.5cm, yshift=0.3cm]{};
							\draw [-, black, dashed] (1,2) -- (0.2,1) node [blue!80!black, xshift=0.5cm, yshift=0.3cm]{};
							\draw [-, black, dashed] (1,2) -- (0.8,1) node [blue!80!black, xshift=0.0cm, yshift=0.3cm]{\small $S_1=s$};
							\draw [-, black, dashed] (1,2) -- (1.2,1) node [blue!80!black, xshift=0.5cm, yshift=0.3cm]{};
							\draw [-, black, dashed] (1,2) -- (1.8,1) node [blue!80!black, xshift=0.5cm, yshift=0.3cm]{};
							\draw [-, black, dashed] (0,1) -- (0,0) node [blue!80!black, xshift=0cm, yshift=-0.3cm]{};
							\draw [-, black, thin] (0.8,1) -- (0,0) node [blue!80!black, xshift=1.1cm, yshift=2.4cm]{\small $q(1-s)$} node [blue!80!black, xshift=0.8cm, yshift=1cm]{\small $S_2=0$};
							\draw [-, black, thin] (0.8,1) -- (0.8,0) node [blue!80!black, xshift=-0.55cm, yshift=2.4cm]{\small $(1-q)$} node [blue!80!black, xshift=0.1cm, yshift=1cm]{\small $S_2=\phi$};
							\draw [-, black, thin] (0.8,1) -- (2,0) node [blue!80!black, xshift=-3.3cm, yshift=2.4cm]{\small $qs$} node [blue!80!black, xshift=-1.2cm, yshift=1cm]{\small $S_2=1$};
							\draw [-, black, dashed] (1.5,1) -- (0,0) node [blue!80!black, xshift=0cm, yshift=-0.3cm]{};
							\draw [-, black, dashed] (1.5,1) -- (1.5,0) node [blue!80!black, xshift=0cm, yshift=-0.3cm]{};
							\draw [-, black, dashed] (1.5,1) -- (2,0) node [blue!80!black, xshift=0cm, yshift=-0.3cm]{};
							\draw [-, black, dashed] (2,1) -- (2,0) node [blue!80!black, xshift=0cm, yshift=-0.3cm]{};
							
							\filldraw[red](1,2) circle (0.5pt) node [blue!80!black, xshift=0cm, yshift=0.2cm]{\small $\pi$};
							\filldraw[red](0,0) circle (0.5pt) node [blue!80!black, xshift=0cm, yshift=-0.3cm]{\small atom at $0$};
							\filldraw[red](2,0) circle (0.5pt) node [blue!80!black, xshift=0cm, yshift=-0.3cm]{\small atom at $1$};
						\end{tikzpicture}
					\end{center}

					What is the optimal persuasion policy in this example? Suppose that in the optimal policy a probability mass of size $\alpha \le 2(1-l)$ is persuaded in period 1. That is, $|\nu_1|=\alpha$ and $\overline{\nu_1}=l$. By Theorem 1, this probability mass must be taken from the interval $[l-\alpha/2, l+\alpha/2]$. Let $\hat{x}(\alpha)$ denote the optimal cutoff such that all probability mass remaining in period 2 with $X_2 \ge \hat{x}(\alpha) $ will be persuaded in period 2. With some algebra, one can show that $\hat{x}(\alpha)$ is given by
					\begin{equation*}
						\hat{x}(\alpha) =l-\sqrt{\frac{(1-l) ^{2}+(1-l) lq ( 1-2\alpha) }{1-q}}
					\end{equation*}
					
					Again we normalize the sender's utility by setting $w_1=1$. Then we can write the sender's utility as a function of $\alpha$:
					\[\Gamma(\alpha)=\alpha + w_2 \left[(1-q)(1-\alpha-\hat{x}(\alpha))+\frac{1}{2}q(1-2l\alpha) \right],\]
					where the first term in the bracket captures the uniform probability mass that is truncated below by $\hat{x}(\alpha)$ and is persuaded when $S_2=\phi$, and the second term represents the atom at $X_2=1$ formed when $S_2=1$. Both terms take into account the fact that the probability mass of $\alpha$ is persuaded in period 1. It is easy to verify that $\Gamma(\alpha)$ is concave in $\alpha$ for all feasible $\alpha \in [0,2(1-l)]$.
					
					The optimal policy is characterized a pair of cutoff functions, $w_L(q,l)$ and $w_H(q,l)$ with $w_L(q,l) < w_H(q,l)$. If $w_2 \le w_L(q,l)$, $\Gamma(\alpha)$ is increasing for all $\alpha \in [0,2(1-l)]$ and hence the optimal policy is greedy in both periods with $\alpha^*=2(1-l)$. If $w_2 \ge w_H(q,l)$, $\Gamma(\alpha)$ is decreasing for all $\alpha \in [0,2(1-l)]$ and hence the optimal policy persuades only in period 2 with $\alpha^*=0$. Finally, if $w_2 \in (w_L(q,l), w_H(q,l))$, an interior $\alpha^* \in (0,2(1-l))$ is optimal and varies with $q$. Interestingly, for a typical fixed $l$, as $q$ increases from $0$, optimal $\alpha^*$ is first $2(1-l)$, 
					then interior, then $0$, then interior again, and finally $2(1-l)$ again as $q$ approaches $1$, as shown in the following figure (where we take $l=3/4,w_2=24/25$):
					
					\includegraphics[scale=0.4]{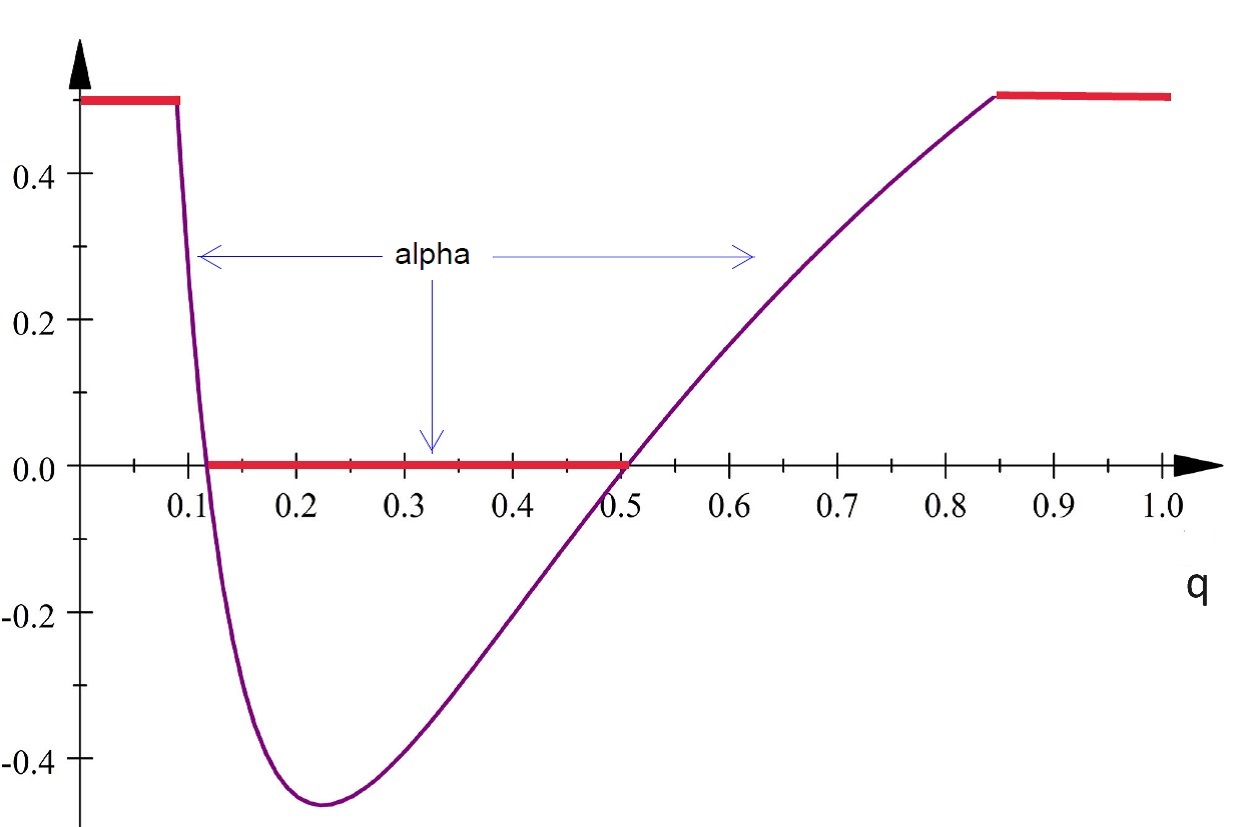}
					
				\end{example}
				
				The following example demonstrates the necessity of the Blackwell order preserving property in Theorem \ref{theo:main}. Namely, we illustrate an example of a two-period model where the kernel is not Blackwell-order preserving and the optimal policy is not an interval policy. 
				
				\begin{example} 
					Intuitively, if for some posteriors much information will be revealed tomorrow while for other posteriors little information will be revealed, then the sender may want to persuade at the posteriors with little future information. For general martingales, the set of posteriors with little information might not form an interval. 
					
					Formally, we consider a two-period interaction with a sender's utility being $w_1=1$ and $w_2=\frac{3}{4}$. 
					Consider the case where \(\mu_1=\frac{1}{7} \delta_{\frac{1}{3}}+ \frac{2}{7} \delta_{\frac{1}{2}}+ \frac{4}{7}\delta_{\frac{3}{4}}\) and \(l=\frac{2}{3}\). The kernel \(\sigma\) is defined as follows:
					\[
					\sigma(\frac{1}{2})=\frac{1}{2}\delta_1+\frac{1}{2}\delta_0, \quad \sigma(\frac{1}{3})=\delta_{\frac{1}{3}}, \quad \text{and} \quad \sigma(\frac{3}{4})=\delta_{\frac{3}{4}}.
					\]
					We note that for \(\nu=\frac{2}{7} \delta_{\frac{1}{2}}+ \frac{4}{7}\delta_{\frac{3}{4}}\) it holds that \(\overline{\nu}=l\).
					Therefore, a policy $\nu'$ is an interval policy iff $\nu':=\nu^\beta=(1-\beta) \nu$ for some
					$\beta\in[0,1]$. In this case $\mu^\beta_2:=\sigma\circ(\mu_1-(1-\beta) \nu)=\beta\frac{1}{7} \delta_{0}+\frac{1}{7} \delta_{\frac{1}{3}}+ \beta\frac{4}{7}\delta_{\frac{3}{4}}+\beta\frac{1}{7} \delta_{1}$. We note that applying the greedy policy to $\mu^\beta_2$ gives the measure
					$\nu^\beta_2:=\beta(\frac{1}{14} \delta_{0}+\frac{1}{7} \delta_{\frac{1}{3}}+ \frac{4}{7}\delta_{\frac{3}{4}}+\frac{1}{7} \delta_{14})$. Overall applying $\nu^\beta$ yields a utility:
					$$(\frac{1}{4}(1-\beta)\frac{6}{7}+\frac{3}{4}((1-\beta)\frac{6}{7}+\beta\frac{13}{14}))=\frac{6}{7}-\frac{9\beta}{56}.$$
					The optimal $\beta$ is therefore $\beta=0$ and the optimal utility for interval policy is $\frac{6}{7}$. 
					
					Consider in contrast an alternative policy of $\nu=\frac{1}{7} \delta_{\frac{1}{3}}+ \frac{4}{7}\delta_{\frac{3}{4}}$. This gives in the next period distribution $\mu_2:=\sigma\circ(\mu_1- \nu)=\frac{1}{7} \delta_{0}+\frac{1}{7} \delta_{1}$ with $\nu_2=\frac{1}{14} \delta_{0}+\frac{1}{7} \delta_{1}$. The utility of this policy is $$(\frac{1}{4}\cdot\frac{5}{7}+\frac{3}{4}(\frac{13}{14}))=\frac{7}{8}>\frac{6}{7}.
					$$
					This implies that any interval policy is sub-optimal for this problem.
					
				\end{example}

				
					

				\subsection{Proof of Theorem \ref{theo:main}}
				
				Our starting point for the proof of Theorem \ref{theo:main} is the following simple claim:
				
				\begin{lemma}\label{lem:sim}
					If $\nu,\nu'\leq\mu_1$ are two measures such that $|\nu| = |\nu'|$, $\overline{\nu}=\overline{\nu'}$, and  $\nu\preceq_B {\nu'}$, then 
					$\mu_1-\nu'\preceq_B\mu_1-\nu$.
				\end{lemma}
				\begin{proof}
					This simple observation follows 
					from the fact that $$\int_0^x\nu([0,y])dy\geq\int_0^x \nu'([0,y])dy$$ for every $x\in[0,1]$, which implies 
					$$\int_0^x (\mu_1-\nu)([0,y]) dy\leq\int_0^x (\mu_1-\nu')([0,y])dy$$ for every $x\in[0,1]$.
				\end{proof}
				\begin{proof}[\textbf{Proof of Theorem \ref{theo:main}}]
					Consider a sender strategy $\{\nu'_t,\mu'_t\}_{t=1,\ldots, T}$ that satisfies the recursive formulation of Lemma \ref{lem:altern}. We first show that there exists an interval policy $\{\nu_t,\mu_t\}_{t=1,\ldots, T}$ that also satisfies the conditions of Lemma   \ref{lem:altern}
					such that $|\nu_t|=|\nu'_t|$. To see this we will use induction and prove that for every \( t \leq T \) it holds that if \( \{\nu_j,\mu_j\}_{j=1,\ldots,t-1} \) satisfy \( |\nu_j|=|\nu'_j| \) for every \( j\leq t-1 \) and \( \mu'_t\preceq_B \mu_t \), then there exists an interval policy \(\nu_t\) with \( |\nu_t|=|\nu'_t| \) such that \( \nu_t\preceq_B \mu_t \) and \( \mu'_{t+1}\preceq_B \mu_{t+1} \). This readily follows from a repeated application of Lemmas \ref{lem:exist}, \ref{lem:obs}, and \ref{lem:sim}. 
					
					For $t=1$ it follows from Lemma \ref{lem:exist} that there exists an interval measure \( \nu_1\) with \( |\nu_1|=|\nu'_1| \) and \( \nu_1 \leq\mu_1 \). Lemma \ref{lem:obs} implies that \( \nu_1 \preceq_B\nu'_1 \). Lemma \ref{lem:sim}
					implies that $\mu'_1-\nu'_1 \preceq_B \mu_1-\nu_1$ (note that $\mu_1=\mu'_1$). Since $\sigma_1$ is Blackwell order preserving it follows that $\mu'_{2}=\sigma_1\circ (\mu'_1-\nu'_1)\preceq_B\sigma_1\circ (\mu_1-\nu_1)=\mu_{2}$ as desired.
					
					Assume the hypothesis holds for \( t-1 \leq T-1 \). We have that $\mu'_t\preceq_B\mu_t$. Therefore there exists a probability kernel \( \kappa :[0,1]\to\Delta([0,1]) \) such that $\kappa \circ\mu'_t=\mu_t$. Let $\psi_t=\kappa \circ\nu'_t$. It holds that $|\psi_t|=|\nu'_t|$ and $\overline{\psi}_t=\overline{\nu'}_t=l$.
					Since  $\mu_t-\psi_t=\kappa \circ(\mu'_t-\nu'_t)$ it holds that 
					\begin{equation}\label{Thm1 proof Eq.1}
						\mu'_t-\nu'_t \preceq_B \mu_t-\psi_t.
					\end{equation}
					By Lemma \ref{lem:exist} there exists an interval policy $\nu_t\leq\mu_t$ such that  
					$|\nu_t|=|\psi_t|$, $\nu_t \preceq_B \psi_t$, and $\overline{\nu}_t=l$. Lemma \ref{lem:sim} therefore implies that 
					\begin{equation}\label{Thm1 proof Eq.2}
						\mu_t-\psi_t\preceq_B \mu_t-\nu_t
					\end{equation}
					Eqs.\ \eqref{Thm1 proof Eq.1} and \eqref{Thm1 proof Eq.2} together with the fact that $\sigma_t$ is Blackwell order preserving yield that $\mu'_{t+1}=\sigma_t\circ (\mu'_t-\nu'_t)\preceq_B\sigma_t\circ (\mu_t-\nu_t)=\mu_{t+1}$ thus concluding the induction step. 
					
					This shows that for every policy there exists an internal policy that achieves at least the same utility for the sender. This implies that the $\sup$ in Equation \eqref{eq:sender-u} can be taken over interval policies to get the optimal utility for the sender.
					The fact that the $\sup$ can be replaced by $\max$ follows from Proposition \ref{prop:existence} in Appendix \ref{sec:max-proof} showing that an optimal interval policy always exists (even for $T=\infty$). 
				\end{proof}
				\section{Application to Random Walks}
				\label{sec:transparancy}
				
				In this section, we restrict attention to information that is revealed to the sender according to a random walk. 
				We recall the notions of the grid points $\Gamma=\{z_j\}_{j\in Z}$ where $Z=\mathbb{Z}\cap[a,b]$ where $a$ can be either finite or $a=-\infty$ and $b$ can be either finite or $b=\infty$. We assume that $z_j<z_{j+1}$ for any $j\in \mathbb{Z}.$
				We normalize (by renaming the indexes) the threshold $l$ to satisfy $l\in (z_0,z_1]$.
				We define
				\begin{align}\label{eq:D}
					D = D(\Gamma) = \sup_{j\leq 0}\frac{l-z_{j-1}}{l-z_{j}}.   
				\end{align}
					Namely, $D\geq 1$ measures the factor by which the distance from the threshold $l$ may grow in a single down-word jump in the random walk, whenever the latter is located strictly below $l$.
				
				We study the case where  $w_t=\delta^{t}$ for some discount factor $\delta>0$ and $X_0=X_1\sim\delta_{z_j}$ for some grid point $z_j\in \Gamma$. That is, the first-day information of the sender is distributed
				as the Dirac measure on $z_j$. The time duration $T$ can be either finite or infinite.
				\begin{theorem}\label{theorem:greedy}
					For a sender whose martingale is a random walk on a grid $\Gamma$ and discounted utilities $w_t=\delta^t$ for $\delta \leq \frac{1}{D(\Gamma)}$ the greedy policy is optimal.
				\end{theorem}
				
				The theorem states that for a sufficiently impatient sender (i.e., with a discount factor below $1/D$), the greedy policy is optimal. We first discuss the tight connection between the greedy policy and transparency. Thereafter, in Section \ref{sec:ex-trm2} we calculate $1/D$ in several examples. In Section \ref{sec:sub-optimal} we demonstrate an example in which the greedy policy is sub-optimal for a sufficiently large discount factor in the case where $T = \infty$.  Section \ref{sec:proof-idea} discusses the ideas of the proof of Theorem \ref{theorem:greedy}, while the formal proof appears in Section \ref{sec:proof-trm2}.

				\paragraph{The transparency of the greedy policy}
				As the greedy policy is defined (see Section \ref{sec:greedy}) the sender stays mute until she makes a recommendation of adoption, while by transparency we relate to the opposite extreme in which the sender reveals her private information in each period. Despite these seemingly two opposite extremes, the arguments below show that in cases in which the martingale is a random walk, there exists a policy that is equivalent (in terms of adoption) to the greedy policy and is almost fully transparent.
				
				Consider the policy that fully reveals the posterior $x_t$ (i.e., reveals the entire private information) as far as $x_t=z_j$ for $j<0$. Once her posterior reaches $x_t=z_{-1}$, at time $t+1$ she either reveals that $x_{t+1}=z_{-2}$ (with a weakly lower probability than it actually happens) or, alternatively, makes a recommendation of adoption to induce the posterior $l$.
				It is easy to see that in this policy the event of adoption is identical to the event of adoption under the greedy policy. But unlike the greedy policy, here the sender stays fully transparent about her private information in all periods until it becomes very close to the threshold (i.e., at the moment the threshold reaches $x_t=z_{-1}$).
				
				\subsection{The value of $1/D$ in Theorem \ref{theorem:greedy}}\label{sec:ex-trm2}
				
				Below we calculate in several examples the threshold $1/D$ for the discount factor below which the greedy policy is optimal by Theorem \ref{theorem:greedy}. In all these examples we have $\frac{1}{D}=\frac{l-z_{0}}{l-z_{-1}}$; namely, the supremum of Equation $\eqref{eq:D}$ is obtained at $j=0$.
				
				\paragraph{Standard grid.} For every $\epsilon>0$ let $\Gamma=\{n\epsilon: n\in \mathbb{N}_0, n\epsilon \leq 1 \}$ be the $\epsilon$-grid. If the threshold is located on the grid ($l\in \Gamma$) then $D=2$ and the maximum is obtained at $j=-1$. If however, the point $l$ is located out of the $\epsilon$-grid, say $l=n\epsilon+\epsilon'$ for $\epsilon'<\epsilon$ the bound obtained by Theorem \ref{theorem:greedy} is worse and equals $\frac{1}{D}=\frac{\epsilon'}{\epsilon'+\epsilon}$ and again is obtained at the value $i=n\epsilon$.
				
				\paragraph{Conditionally i.i.d. binary symmetric signals.} Consider a scenario in which the sender observes in each period a signal that equals the state $\omega$ with probability $p\in (\frac{1}{2},1)$ independently across periods (conditional on the states). For simplicity of the calculation let us assume that $\frac{1}{2}\in \Gamma$. In such a case the grid $\Gamma$ takes the simple form of $\Gamma=\{\frac{p^z}{p^z+(1-p)^z}:z\in \mathbb{Z}\}$. Assume that the initial prior $\pi<\frac{1}{2}$ and that $l=\frac{1}{2}$. In this case, again, the supremum of Equation \eqref{eq:D} is obtained at the grid point left to the threshold $l=\frac{1}{2}$ and  $$\frac{1}{D}=\frac{\frac{1}{2}-\frac{p^{-1}}{p^{-1}+(1-p)^{-1}}}{\frac{1}{2}-\frac{p^{-2}}{p^{-2}+(1-p)^{-2}}}=2p^2-2p+1.$$
				
				Notice that as the signals of the sender become more accurate ($p\to 1$) the greedy policy becomes optimal for discount factors that approach 1; namely $\lim_{p\to 1} \frac{1}{D}=1$.
				
				\paragraph{A standard grid with a hole} Another class of grids for which Theorem \ref{theorem:greedy} proves the optimality of the greedy policy for an arbitrary patient sender is the can where the standard grid omits grid points that are located close to $l$ but from the left of $l$. Let $\epsilon'>>\epsilon$ and let $\Gamma=\{n\epsilon: n\in \mathbb{N}\cap \{0\}, n\epsilon\leq 1, n\epsilon \notin (l-\epsilon',l)\}$. Namely, from the point $x_t=l-\epsilon'$ the belief of the sender either moves to the nearby point $x_{t+1}=l-\epsilon'-\epsilon$ or jumps to the far point which (weekly) exceeds $l$. In such a case $\frac{1}{D}\geq \frac{\epsilon'}{\epsilon'+\epsilon}$ which approaches 1 as $\frac{\epsilon}{\epsilon'} \to 0$.

				\subsection{Idea of the proof of Theorem \ref{theorem:greedy}}\label{sec:proof-idea}
				
				We start with the case $T<\infty$ and use Backward-induction. By the single deviation principle (or equivalently by the backward induction hypothesis) it is sufficient to prove that for every state $X_0\in \Delta(\Gamma)$ if starting from tomorrow the sender acts greedily, she will be better off acting greedily today as well. 
				
				Theorem \ref{theo:main} allows us to restrict the set of possible policies significantly: the decision in each state $X_1$ is characterized by an interval $[\underline q,\overline q]$ such that the conditional mean of $X_1$ over $[\underline q,\overline q]$ is $l$. Namely, there is no need to consider the set of all sub-measures whose mean is $l$ but only those that are supported on an interval.
				
				The proof that acting greedily today is indeed superior is done in two steps. First, we show that the utility of acting greedily in the first two days provides the best possible utility for the first two periods among all possible interval policies; see Lemma \ref{prop:inf} and Corollary \ref{cor:greedy}.
				Here we use the assumption that $\delta\leq \frac{1}{D}$. \footnote{Without this assumption this claim is false as can be seen in the proof of Proposition \ref{pro:sub-opt}.}
				
				In the second step, we prove that the remaining mass after two consecutive applications of the greedy action is superior to the remaining mass after any interval action followed by a single greedy action; see Lemma \ref{lemma:gt}. Here we shall clarify in which sense it is superior. The second-order stochastic dominance partial order is irrelevant here because we compare two measures with different masses. The first-order stochastic dominance (FOSD) is indeed satisfied, but there is an obstacle: the FOSD order is not necessarily preserved under the random walk over a grid. Therefore, a weaker (than FOSD) version of dominance is needed for our arguments; see Definition \ref{def:domination} which is related to to the notion known in the literature as the \emph{increasing convex order} (see, e.g., \citep{shaked2007stochastic}). Out notion of domination seem to extend this notion to general positive finite measure.
				
				This order exactly serves our purposes. On the one hand it is preserved by a random walk on a grid; see Lemma \ref{lemma:dominPreserv}. On the other hand, it is sufficiently powerful to deduce that the sender will be better off by remaining with the dominant measure; see Lemma \ref{lemma:domination}.
				
				To summarize, we show that the greedy policy is superior to any other interval policy in both aspects: it provides better utility in the first two periods (the first step above) and it leaves the sender with a measure that she can better utilize in the future periods (the second step above).
				
				\subsubsection{Extension to initially partially-informed sender setting}
				
				Note that our proof shows the optimality of the greedy policy not only for the case in which the sender is initially uninformed (i.e., $X_0=\delta_\pi$) but also for the case in which the partial initial information of the sender that might be captured by an arbitrary distribution $X_1\in \Delta(\Gamma)$ with $\Ex[X_1] \leq l$. 
				In Section \ref{sec:proof-trm2} we define the set of measures $\Delta^*(\Gamma)$ that comprises all positive measures that are supported on the even or odd points of the greed. 
				The following corollary follows from the proof of Theorem \ref{theorem:greedy}.
				\begin{corollary}\label{cor:partially}
					If the sender's private information martingal follows a random walk on the grid $\Gamma$ with initial distribution $X_1\sim \Delta^*(\Gamma)$ and her discount factor satisfies $\delta\leq \frac{1}{D}$ then the greedy policy is optimal. 
				\end{corollary}

				This simply follows from the fact that the single deviation principle (Bellman equations) technique indeed shows the optimality of the greedy for every initial distribution.
				
				\subsection{Sub-optimality of the greedy policy for a patient sender}\label{sec:sub-optimal}
				
				Whether the greedy policy is optimal for an initially uninformed sender whose martingale of beliefs follows a random walk on a grid (either the standard grid or the one induced from conditionally i.i.d. binary signals) remains a major open problem. However, in this subsection, we show that the stronger claim of Corollary \ref{cor:partially} is invalid for sufficiently large discount factors. In other words, if we are interested in proving the optimality of the greedy policy for any initially partially informed sender, some restrictions on the patience of the sender are unavoidable. This observation, in particular, indicates the challenges of resolving the above-mentioned open problem; The single-deviation principle (Bellman equations) is a central tool for proving results in such dynamic settings. Our negative observation indicates that different tools will be needed to prove such a result if greedy is indeed the optimal policy. 
				
				We consider the standard $\epsilon$-grid; i.e., $\Gamma=\{n\epsilon: n \in \mathbb{N}_0, n\epsilon<1\}$.
				
				\begin{restatable}{proposition}{firstprop}\label{pro:sub-opt}
					Consider the case where $T=\infty$.   For every $\delta>\frac{1}{\sqrt{2}}$ there exists an $\epsilon'$ such that for every $\epsilon$-grid $\Gamma$ that contains the threshold $l\in \Gamma$ for $\epsilon<\epsilon'$ there exists an initial prior $X_0\in \Delta(\Gamma)$ for which the greedy policy is sub-optimal.
				\end{restatable}
				The proof follows the following steps. We first assume by way of contradiction that the greedy policy is always optimal and denote by $v(j\epsilon)$ the sender's value from the prior $\delta_{j\epsilon}$ for any grid point $j\epsilon$. Using coupling considerations, we show that since the greedy policy is optimal, $v(l-2\epsilon) \approx (v(l-\epsilon))^2$. We use this to approximate $v(l-\epsilon)$ as a function of $\delta$ and $\epsilon$. We then consider a prior that is supported on the points $l-2\epsilon$ and $1$, and has an expectation below $l$. Since using the greedy policy leaves, at the next stage, only mass on the point $l-2\epsilon$, we use the above approximation to estimate the value of the greedy policy. We then show that waiting a single period without using any mass and then pooling the mass obtained at $l-\epsilon$ together with the mass at $1$ gives a better utility for the sender. The formal proof is relegated to Appendix \ref{ap:proof_pro_1}.
				
				In light of Proposition \ref{pro:sub-opt} an interesting question that we leave open is whether the bound in Theorem \ref{theorem:greedy} is tight.


				\subsection{Proof of Theorem \ref{theorem:greedy}}\label{sec:proof-trm2}
				Recall that given two positive measures $\phi, \mu$ on $[0,1]$ we denote $\phi\leq\mu$ if $\phi(B)\leq\mu(B)$ for any Borel measurable set $B\subseteq[0,1]$. 
				We now introduce relations between positive measures on $[0,1]$ that play a fundamental role in our analysis. The first is the well-known first order stochastic domination. A measure $\mu$ first order stochastic dominates (FOSD) a measure $\lambda$ ($\lambda\preceq_F \mu$) if $\mu([x,1])\geq\lambda([x,1])$ for any $x\in[0,1]$. A function $\rho:[0,1]\to\Delta([0,1])$ is called a FOSD kernel  if $\varphi(x)([x,1])=1$. It is easy to see that 
				$\lambda\preceq_F\mu$ iff there exists a FOSD kernel $\varphi$ and a measure $\phi\leq \mu$ such that $\varphi\circ\lambda=\phi$. 
				
				A central notion of our analysis is a notion that we call \emph{domination}. 
				\begin{definition}\label{def:domination}
					Say that a  measure $\mu$ \emph{dominates} a measure $\lambda$ ($\lambda\preceq_D \mu$) if there exists a FOSD kernel $\varphi:[0,1]\to\Delta([0,1])$ and a probability kernel $\rho:[0,1]\to\Delta([0,1])$ such that there exists a measure $\phi\leq\mu$ such that 
					$\rho\circ\varphi\circ\lambda=\phi$.    
				\end{definition}
				It follows directly from the definition that the notion of domination extends first and second-order stochastic domination. Namely, $\lambda\preceq_F \mu \Rightarrow \lambda\preceq_D \mu$ and $\lambda\preceq_B \mu \Rightarrow \lambda\preceq_D \mu$. In the special case where $|\lambda|=|\mu|$, the domination order is called in the \emph{increasing convex order} in the literature; see \citep{shaked2007stochastic}.\footnote{We omit the proof of the observation that our notion of domination is equivalent to the increasing convex order because our proof does not rely on this observation.}

				We start with the first lemma for proving our theorem.
				\begin{lemma}\label{lemma:gt}
					Consider an infinite Blackwell-order preserving martingale $\bX=(X_t)_{t\geq 1}$. Consider two interval policies $(\nu_1,\nu_2)$ and $(\nu'_1,\nu'_2)$ for the two first periods such that $(\nu_1,\nu_2)$ is the greedy policy and $(\nu'_1,\nu'_2)$ is any other policy where $\nu'_2$ is greedy. 
					Then $\mu'_2-\nu'_2\preceq_F\mu_2-\nu_2$.
				\end{lemma}
				\begin{proof}
					Note first that by definition  $\mu_2=\sigma_1\circ(\mu_1-\nu_1)$ and $\mu'_2=\sigma_1\circ(\mu_1-\nu'_1)$.
					Since $\nu_1$ is the greedy policy and since $\nu'_1$ is an interval policy it follows that $\nu'_1\leq \nu_1$. Therefore, we can write $\nu_1=\nu'_1+\psi$ for some positive measure $\psi$. Note further that since $\overline{\nu}_1=\overline{\nu'}_1=l$ we must also have that $\overline{\psi}=l$.
					Therefore, since
					\begin{align*}
						\mu_2 = \sigma_1 \circ (\mu_1 - \nu_1) = \sigma_1 \circ (\mu_1 - \nu'_1-\psi) = \mu'_2 - \sigma_1 \circ \psi,
					\end{align*}
					we infer that that $\mu'_2= \mu_2+\phi$ for some positive measure $\phi$ satisfying $\overline{\phi}=l$ (by setting $\phi = \sigma_1 \circ \psi$). Next, by introducing $\lambda = \mu_2 - \nu_2$ we may further decompose $\mu'_2$ to 
					\begin{align}\label{Eq. decom. mu'}
						\mu'_2 = \lambda + \nu_2  + \phi.    
					\end{align}
					As $\overline{\nu_2} = l$ and $\overline{\phi}=l$, it follows easily from the linearity of the integral operator that $\overline{\nu_2+\phi}=l$ as well. Therefore, using decomposition \eqref{Eq. decom. mu'} we infer that $\nu_2+\phi$ defines a stopping rule with respect to $\mu'_2$. As $\nu'_2$ is the greedy policy with respect to $\mu'_2$,  $\nu'_2 ([x,1]) \geq (\nu_2+\phi)([x,1])$ for every $x \in [0,1]$. This is because that greedy stopping rule is maximal among all compatible stopping rules in first-order stochastic domination partial order relation. Thus, for every $x \in [0,1]$, \begin{align*}
						&(\mu'_2-\nu'_2)([x,1])=\mu'_2([x,1])-\nu'_2([x,1])\leq \\ 
						& \mu'_2([x,1])-(\nu_2+\phi)([x,1])=(\mu_2-\nu_2)([x,1]),
					\end{align*}
					where the last equality follows from $\mu'_2= \mu_2+\phi$.
					This indeed shows that $\mu_2-\nu_2$ first-order stochastically dominates  $\mu'_2-\nu'_2$, as desired.
				\end{proof}

				
				
				An important property of the domination relation defined above is that it is closed under addition. That is:
				\begin{lemma}\label{lem:aditive}
					Assume that $\psi\preceq_D \mu$ and that $\psi'\preceq_D \mu'$ for some positive measure $\mu,\mu',\psi,\psi'$ on a discrete grid $\Gamma$. Then it holds that  $\psi+\psi'\preceq_D\mu+\mu'.$ 
				\end{lemma}
				\begin{proof}
					By definition there exists $\varphi,\rho$ and $\varphi',\rho'$ such that $\rho\circ\varphi\circ\psi:=\phi\leq\mu$ and $\rho'\circ\varphi'\circ\psi':=\phi'\leq\mu'$. We first show that there exists a FOSD kernel $\tilde\varphi$ satisfying $\tilde \varphi\circ(\psi+\psi')=\varphi\circ\psi+\varphi'\circ\psi'$. Define the FOSD kernel $\tilde\varphi$ as follows: $\tilde \varphi(x)=\frac{\psi(x)}{(\psi+\psi')(x)}\varphi(x)+\frac{\psi'(x)}{(\psi+\psi')(x)}\varphi'(x)$ if $(\psi+\psi')(x)>0$ and  $\tilde \varphi(x)=\delta_x$ otherwise.
					
					We have that 
					\begin{align*}
						& \tilde\varphi\circ(\psi+\psi')(y)=\sum_{x\in \Gamma ,\ (\psi+\psi')(x)>0}(\psi+\psi')(x)\tilde\varphi(x)(y)=\\ &\sum_{x \in \Gamma,\ (\psi+\psi')(x)>0}(\psi+\psi')(x)\Big(\frac{\psi(x)}{(\psi+\psi')(x)}\varphi(x)(y)+\frac{\psi'(x)}{(\psi+\psi')(x)}\varphi'(x)(y)\Big)=\\ &\sum_{x \in \Gamma ,\ (\psi+\psi')(x)>0}\psi(x)\varphi(x)(y)+\psi'(x)\varphi'(x)(y)=\varphi\circ\psi(y)+\varphi'\circ\psi'(y). 
					\end{align*}
					Therefore $\tilde \varphi\circ(\psi+\psi')=\varphi\circ\psi+\varphi'\circ\psi'$.
					
					We next show that $\tilde\rho\circ(\varphi\circ\psi+\varphi'\circ\psi')=\rho\circ\varphi\circ\psi+\rho'\circ\varphi'\circ\psi$ for some probability kernel $\tilde\rho$.
					Similarly, define $$\tilde \rho(x)=\frac{\varphi\circ\psi(x)}{(\varphi\circ\psi+\varphi'\circ\psi')(x)}\rho(x)+\frac{\varphi'\circ\psi'(x)}{(\varphi\circ\psi+\varphi'\circ\psi')(x)}\rho'(x)$$ if $(\varphi\circ\psi+\varphi'\circ\psi')(x)>0$ and  $\tilde \rho(x)=\delta_x$ otherwise.
					A similar calculation as above shows that $\tilde{\rho} \circ (\varphi \circ \psi + \varphi' \circ \psi') = \rho \circ \varphi \circ \psi + \rho' \circ \varphi' \circ \psi'.
					$
					Overall we have shown that $$\tilde{\rho} \circ\tilde{ \varphi}\circ ( \psi + \psi') = \rho \circ \varphi \circ \psi + \rho' \circ \varphi' \circ \psi'=\phi+\phi'\leq\mu+\mu',
					$$
					as desired.
				\end{proof}
				The next lemma shows that the random walk kernel also preserves domination. 
				\begin{lemma}\label{lemma:dominPreserv}
					Let $\sigma$ be a random walk kernel and $\mu,\psi\in\Delta(\Gamma)$ be two positive measures on $\Gamma$ such that $\psi\preceq_D \mu$. Then  $\sigma\circ\psi\preceq_D \sigma\circ\mu$. 
				\end{lemma}
				
				\begin{proof}
					Let  $\varphi$ and $\rho$ such that $\rho\circ\varphi\circ\psi \leq\mu$. 
					We first show that if $y\in \Gamma$ and $x\leq y$ then $\sigma\circ\rho\circ\delta_y$ dominates $\sigma\circ\delta_{x}.$

					Assume $x=y$ then we need to show that $\sigma\circ\rho\circ\delta_x$ dominates $\sigma\circ\delta_{x}.$
					This follows since $\rho\circ\delta_x$ second-order stochastically dominates $\delta_x$ and $\sigma$ is  Blackwell order-preserving and so $\sigma\circ\rho\circ\delta_x$ second-order stochastically dominates $\sigma\circ\delta_{x}.$ Since domination generalizes second-order stochastic domination, the result follows.

					Consider the case where $x<y$. We will show that there exists probability kernel $\tilde\rho$ and FOSD kernel $\tilde\varphi$ such that 
					$\tilde\rho\circ\tilde\varphi\circ(\sigma\circ\delta_x)=\sigma\circ\rho\circ\delta_y.$ Define $\tilde\varphi(x')=\delta_y$ for all $x'<y$ and $\tilde\varphi(x')=\delta_{x'}$ for $x'\geq y$. 
					Let $\tilde\rho=\sigma\circ\rho$. 
					Note that since $x<y$ it holds that the realizations of $\sigma\circ\delta_x$ lie weakly below $y$. Therefore, $\tilde\varphi\circ(\sigma\circ\delta_x)=\delta_y$.
					We have that $\tilde\rho\circ\tilde\varphi\circ(\sigma\circ\delta_x)=\tilde\rho\circ\delta_y=\sigma\circ\rho\circ\delta_y$.
						Hence  $\sigma\circ\delta_{x}\preceq_D\sigma\circ\rho\circ\delta_y.$ 
						
						We are now in position to prove the Lemma. Since $\varphi(x)([x,1])=1$ there must be $\beta_y\geq 0$ for any $y\geq x$ such that $\sum_{y\geq x}\beta_y=1$ and  $\varphi(x)=\sum_{y\geq x}\beta_y\delta_y$. It now follows that $\sigma\circ\rho\circ\varphi\circ\delta_x$ dominates $\sigma\circ\delta_x$. To see this note that $\sigma\circ\rho\circ\varphi\circ\delta_x=\sum_{y\geq x}\beta_y\sigma\circ\rho\circ\delta_y$
						and $\sigma\circ\rho\circ\delta_y$ dominates $\sigma\circ\delta_x$ for any $y\geq x$. Since we can write $\sigma\circ\delta_x=\sum_{y\geq x}\beta_y\sigma\circ\delta_x$ Lemma \ref{lem:aditive} implies that $\sigma\circ\delta_x\preceq_D\sigma\circ\rho\circ\varphi\circ\delta_x$.
						
						Let $\psi=\sum_{x}\alpha_x\delta_x$. Note that $\sigma\circ\psi=\sum_{x}\alpha_x\sigma\circ\delta_x$.
						Let $\phi\leq\mu$ such that $\rho\circ\varphi\circ\psi=\phi$. It holds that 
						$\sigma\circ\phi=\sum_x\alpha_x\sigma\circ\rho\circ\varphi\circ\delta_x$.
						Since $\sigma\circ\delta_x\preceq_D\sigma\circ\rho\circ\varphi\circ\delta_x$ it follows that from Lemma \ref{lem:aditive} that 
						$\sigma\circ\psi\preceq_D\sigma\circ\phi$. Since $\sigma\circ\phi\leq \sigma\circ\mu$ it also follows that 
						$\sigma\circ\psi\preceq_D\sigma\circ\mu$ as desired. 
						

						
					\end{proof}
					The following lemma shows that domination plays in favor of the sender. 
					\begin{lemma}\label{lemma:domination}
						Let $\bX=(X_t)_{t=1,\ldots}$ and $\bX'=(X'_t)_{t=1,\ldots}$ be two random walks on $\Gamma$ with the same kernel $\sigma$. Assuming that $\mu_1$ dominates $\mu'_1$, then the optimal policy under $\bX$ yields a higher payoff for the sender than the optimal policy under $\bX'$.   
					\end{lemma}
					\begin{proof}
						The proof uses similar considerations as the proof of Theorem \ref{theo:main}. Let $\{\nu'_t,\mu'_t\}$ be an optimal interval strategy of the sender for $\bX'$. We will show that if $\mu_t$ dominates $\mu'_t$ then there exists $\nu_t$ such that $|\nu_t|=|\nu'_t|$, $\overline{\nu}_t\geq l$, and $\mu_t-\nu_t$ dominates $\mu'_t-\nu'_t$. By definition of domination there exists $\varphi$, $\rho$, and $\phi\leq \mu_t$ such that $\varphi\circ\rho\circ\mu'_t=\phi\leq\mu_t$. Consider $\nu_t=\varphi\circ\rho\circ\nu'_t$. It holds 
						that $\nu_t\leq\phi\leq \mu_t$. Since $\varphi$ is a FOSD kernel and $\rho$ is a probability kernel, it holds that $\overline{\nu}_t\geq l$ and $|\nu_t|=|\nu'_t|$.
						Finally, we have that by construction that $$\varphi\circ\rho\circ(\mu'_t-\nu'_t)=\varphi\circ\rho\circ\mu'_t-\varphi\circ\rho\circ\nu'_t=\phi-\nu_t\leq \mu_t-\nu_t.$$
						Therefore $\mu_t-\nu_t$ dominates $\mu'_t-\nu'_t$ as desired. It follows from Lemma \ref{lemma:dominPreserv} that $\mu_{t+1}=\sigma\circ(\mu_t-\nu_t)$ dominates $\mu'_{t+1}=\sigma\circ(\mu'_t-\nu'_t)$. Therefore we can construct a policy $\{\mu_t,\nu_t\}_{t=1,\ldots}$ such that $|\nu_t|=|\nu'_t|$ and $\overline{\nu}_t\geq l$. This concludes the proof of the lemma.
					\end{proof}    
					
					We call a grid point $z_j$ for $a<j<b$ odd (even) if $j$ is odd (even). Note that for $j\in\{a,b\}$ such that $j\neq\infty,-\infty$ the point $z_j$ will be defined to be both odd and even.
					We denote by $\Delta^*(\Gamma)$ the set of positive finite measures that are supported on  either the even or odd points of the grid. Note that if $\mu\in\Delta^*(\Gamma)$, then $\sigma\circ\mu\in\Delta^*(\Gamma)$.  
					
					The following lemma provides an upper bound on how much can the sender improve her persuasion mass by waiting one additional period. For any measure $\mu\in\Delta(\Gamma)$ let $\nu$ be the greedy policy with respect to $\mu$, and set  $g(\mu):=|\nu|$.
					
					\begin{lemma}\label{prop:inf}
						For any measure $\mu\in \Delta^*(\Gamma)$ and every $\delta\leq\frac{1}{D}$ (see Equation \eqref{eq:D} for the definition of $D$) we have
						$$ \delta g(\sigma\circ\mu) \leq g(\mu)+\delta g(\sigma\circ(\mu-\nu)).$$
					\end{lemma}
					This essentially says that in a two-period discount factor model if the discount factor satisfies $\delta\leq\frac{1}{D}$, then the greedy policy is better than waiting and not revealing any information at the first period and then applying the greedy in the second.

					\begin{proof}
						We first consider the case where 
						the lowest point in the support of the greedy measure $\nu\leq\mu$ is a point $z_j<z_0$. 
						In this case, it must hold in particular that $g(\sigma\circ(\mu-\nu))=0$. We therefore need to show that $\frac{g(\sigma\circ\mu)}{g(\mu)}\leq D$. We next approximate $g(\sigma\circ\mu)$. 
						

						Let $\nu'$ be the greedy measure for $\sigma\circ\mu$. 
						We consider two cases. The first is that the lowest point in the support of $\nu'$ is $z_{j-1}$. In this case, it must hold that $\nu'=\sigma\circ\nu$ since $\overline{\sigma\circ\nu}=l$ on the one hand and on the other hand all the mass that lies strictly above $z_{j-1}$ is contained in $\sigma\circ\nu.$
						
						Consider the case where the lowest point in the support of $\nu'$ is not $z_{j-1}$. In this case, it cannot lie below $z_{j-1}$ since the conditional expectation of such a measure lies strictly below $l$.  Since $\sigma\circ\mu$ is supported either on the odd or the even points, the lowest point in the support of $\nu'$ must be $z_{j+1}$. 
						
						We denote by $\lambda$ the submeasure of $\nu$ containing all mass of points lying strictly above $z_j$. We note that $\nu=\alpha\delta_{z_j}+\lambda $ for some $\alpha>0.$ Since $\overline{\nu}=l$ it holds that  $$\alpha=\frac{\overline{\lambda}-l}{l-z_j}|\lambda|.$$
						
						We note that one can write $\nu'=\beta\delta_{z_{j+1}}+\sigma\circ\lambda$ for some $\beta>0$. To see this note that the lowest point in the support of $\sigma\circ\lambda$ lies (weakly) above $j+1$ and that $\overline{\sigma\circ\lambda}=\overline{\lambda}>l$, therefore, $\sigma\circ \lambda\leq \nu'$.
						Again since $\overline{\nu'}=l$, we must have that
						$$\beta=\frac{\overline{\sigma\circ\lambda}-l}{l-z_{j+1}}|\sigma\circ\lambda|=\frac{\overline{\lambda}-l}{l-z_{j+1}}|\lambda|.$$
						
						Therefore, it holds that
						
						$$\frac{g(\sigma\circ\mu)}{g(\mu)} = \frac{|\nu'|}{|\nu|}\leq\frac{\beta}{\alpha}=\frac{l-z_{j+1}}{l-z_j}\leq D.$$

						The second case to consider is where the left-most point in $\nu$ is $z_0$. We claim that in this we must have that $g(\mu)+g(\sigma\circ(\mu-\nu))=g(\sigma\circ\mu)$. 
						We first show that $\sigma\circ\nu\leq \nu'$. Let $\lambda$ be the submeasure that contains all points that lie strictly above $z_0$ in $\mu$. It must hold that $\sigma\circ \lambda\leq \nu'$. In addition it must contain all mass that arrives from $z_0$.  Because this is the right-most mass of $\sigma\circ(\mu-\lambda)$.  Therefore $\sigma\circ\nu\leq \nu'$.

						As a result $\sigma\circ\mu\geq \nu'=\sigma\circ\nu+(\nu'-\sigma\circ\nu)$ where $\nu'-\sigma\circ\nu\geq 0$ is positive and $\overline{\nu'-\sigma\circ\nu}=l$. Therefore $\sigma\circ(\mu-\nu)=\sigma\circ\mu-\sigma\circ\nu\geq \nu'-\sigma\circ\nu$. This means that $g(\sigma\circ(\mu-\nu))\geq |\nu'-\sigma\circ\nu|$ and therefore $g(\mu)+g(\mu-\nu)\geq |\nu|+|\nu'-\sigma\circ\nu|=|\nu'|=g(\sigma\circ\mu)$.
						
						This completes the proof of the lemma. 
					\end{proof}
					
					The following corollary essentially shows that the greedy policy is optimal for the case where $T=2$.
					\begin{corollary}\label{cor:greedy}
						For any measure $\mu\in \Delta^*(\Gamma)$, $\alpha\leq g(\mu)$, and every $\delta\leq\frac{1}{D}$ we have
						$$\alpha+\delta g(\sigma\circ(\mu-\nu_\alpha)) \leq g(\mu)+\delta g(\sigma\circ(\mu-\nu)),$$
						where $\nu_\alpha\leq\mu$ is the measure associated with the unique interval policy for $\mu$ of mass $\alpha$.
					\end{corollary}
					\begin{proof}
						Let $\mu'=\mu-\nu_\alpha$ and denote $\nu'$ to be the greedy measure for $\mu'$.  By the interval property of $\nu_{\alpha}$ we have $\nu = \nu_{\alpha} + \nu'$.
						Therefore, $g(\mu')= |\nu'|=g(\mu)-\alpha$, and $\mu'-\nu'=\mu-\nu$.
						Using Lemma \ref{prop:inf} for the measure $\mu'$, together with the latter two properties we obtain
						\begin{align*}
							\delta(\sigma\circ\mu')& \leq g(\mu')+\delta g(\sigma\circ(\mu'-\nu'))\\
							& = g(\mu)-\alpha +\delta g(\sigma\circ(\mu-\nu)),
						\end{align*}
						thus giving the desired result.
					\end{proof}
					
					We are now ready to prove Theorem \ref{theorem:greedy}.
					\begin{proof}[\textbf{Proof of Theorem \ref{theorem:greedy}}]
						Fix a discount function $\delta\leq\frac{1}{D}.$
						We prove Theorem \ref{theorem:greedy} first for $T<\infty$ under the assumption that the prior $\mu\in \Delta^*(\Gamma)$. Denote by $v_t(\mu)$ be the value of the sender's problem as a function of the prior $\mu$ (i.e., $X_1 \sim \mu$) and the number of periods $t$. Also, denote $\gamma_t(\mu)$ to be the payoff obtained along the first $t$ periods by following the greedy policy starting from the prior $\mu$. 
						
						We prove by induction on $T$ that $v_T(\mu)=\gamma_T(\mu)$.
						For the case $T=1$ the claim trivially holds. Assume the claim holds for any $t<T$. As by Theorem \ref{theo:main} we may restrict attention to interval policies, the induction hypothesis implies
						\begin{align}\label{eq:bi}
							v_T(\mu)=\max_{\alpha\leq g(\mu)}\lbrace \alpha+\delta\gamma_{T-1}(\sigma\circ(\mu-\nu_\alpha)) \rbrace ,   
						\end{align}
						where  $\nu_\alpha$ is the measure associated with the unique interval policy for $\mu$ of mass $\alpha\leq g(\mu)$. 
						
						\bigskip
						
						Fix $\alpha\leq g(\mu)$. Denote $\mu'_2=\sigma\circ(\mu-\nu_\alpha)$ and let $\nu'_2$ be the greedy measure with respect to $\mu'_2$. By the definition of the greedy policy we have, 
						\begin{align}\label{eq:rew}
							\alpha+\delta\gamma_{T-1}(\sigma\circ(\mu-\nu_\alpha)) & = \alpha + \delta \left( g(\mu'_2)+ \gamma_{T-1}(\mu'_2-\nu'_2)\right) \nonumber \\
							& = \alpha + \delta g(\mu'_2)+ \delta v_{T-1}(\mu'_2-\nu'_2)
						\end{align}
						where the last equality follows from the induction hypothesis. 
						
						Denote $\mu_2 = \sigma \circ (\mu - \nu)$, where we recall that $\nu$ is the greedy measure for $\mu$. Let $\nu_2$ be the greedy measure for $\mu_2$. As  Lemma \ref{lemma:dominPreserv} implies that $\mu'_2-\nu'_2 \preceq_D \mu_2 - \nu_2$, we get by Lemma \ref{lemma:domination} that $v_{T-1}(\mu'_2-\nu'_2) \leq v_{T-1}(\mu_2-\nu_2)$. This together with \eqref{eq:rew} and Corollary \ref{cor:greedy} implies that for any $\delta \leq 1/D$:
						\begin{align}\label{eq:rew2}
							\alpha + \delta g(\mu'_2)+ \delta v_{T-1}(\mu'_2-\nu'_2) & \leq g(\mu)+\delta g(\sigma\circ(\mu-\nu)) + \delta v_{T-1}(\mu_2-\nu_2) \nonumber \\
							& =  g(\mu)+\delta g(\sigma\circ(\mu-\nu)) + \delta \gamma_{T-1}(\mu_2-\nu_2)\\
							& = \gamma_T (\mu), \nonumber 
						\end{align}
						where the first equality follows from the induction hypothesis, and the second equality follows from the definition of the greedy policy. Combining eqs.\ (\ref{eq:rew}) and (\ref{eq:rew2}) we obtain that $\alpha+\delta\gamma_{T-1}(\sigma\circ(\mu-\nu_\alpha)) \leq \gamma_T (\mu)$. As $\alpha$ was arbitrary, the latter in turn with \eqref{eq:bi} implies that $v_T(\mu) \leq \gamma_T (\mu)$. As the reverse inequality holds as well, we've completed the induction step, and thus also the proof of Theorem \ref{theorem:greedy} for the case $T<\infty$. 
						\bigskip 
						
						As for the case $T=\infty$, we let $v_\infty(\mu)$ and $\gamma_\infty(\mu)$ denote the value and the greedy payoff for $T=\infty$. Note that $v_t(\mu) \uparrow v_\infty(\mu)$ and $\gamma_t(\mu) \uparrow \gamma_{\infty}(\mu)$ as $t \uparrow \infty$. Therefore it must hold that $\gamma_\infty(\mu)=v_\infty(\mu)$ as otherwise  we would get that $\gamma_t(\mu)<v_t(\mu)$ for some $t<\infty$, arriving at a contradiction.

					\end{proof}
					
					\bibliography{bib}
					
					\appendix
					\section{Ommited Proofs from Section \ref{sec:preliminaries}}\label{ap:sec2}
					This appendix contains all the proofs that have been omitted in the body of the paper.
					
					\firstlem*
					
					\begin{proof}
						
						We first show that any stopping rule defines the above sequences of measures. Given a stopping rule $\tau$ we define for every $t=1,\ldots, T$, the positive measure $\mu_t=\Pro_{X_t}[\ \cdot\ |\tau\geq t]\, \Pro[\tau\geq t]$ describing the unconditional distribution of $X_t$ on the event $\{\tau\geq t\}$. Also, let $\nu_t=\Pro_{X_t}[\ \cdot\ |\tau=t]\, \Pro_{X_t}[\tau=t]$ be the measure corresponding to the unconditional distribution of $X_t$ on the event $\{\tau = t\}$.  Clearly, $\overline{\nu_t}\geq l$. Moreover, it follows from the definitions that  $\mu_t=\sigma_{t-1} \circ (\mu_{t-1} - \nu_{t-1})$ and $\nu_t\leq\mu_t$ for every $t\geq 1$, so that $\{\mu_t\}_{t\geq 1}$ and $\{\nu_t\}_{t\geq 1}$ meet the required constraints. Lastly, note that as $\Pro(\tau= t)=|\nu_t|$ we have $\sum_{t=1}^T\Pro(\tau=t)w_t=\sum_{t=1}^T |\nu_t|w_t$, as desired.
						
						Conversely, for every two sequences of positive measures $\{\mu_t\}_{t\geq 1}$, $\{\nu_t\}_{t\geq 1}$ that satisfy the above relations we will show that there exists a stopping rule $\tau$ such that $\mu_t=\Pro_{X_t}[\ \cdot\ |\tau\geq t]\Pro[\tau\geq t]$ and $\nu_t=\Pro_{X_t}[\ \cdot\ |\tau=t]\, \Pro[\tau=t] $.  
						
						To see this, we define the stopping rule recursively such that $\tau_t$ depends only on the realization $x_t$ of $X_t$. Since $\nu_1\leq \mu_1$ the Radon Nikodym derivative $\frac{d\nu_1}{d\mu_1} : [0,1] \to \R$ satisfies $\frac{d\nu_1}{d\mu_1}(x)\leq 1$ for $\mu_1$ almost every $x\in[0,1]$. By setting $\tau_1 := \frac{d\nu_1}{d\mu_1}$, we have $\nu_1 = \Pro_{X_1}[\ \cdot\ |\tau=1]\, \Pro[\tau=1]$ as required. We proceed inductively. Assume that we have defined $(\tau_j)_{j=1,\ldots,t-1}$ such that $\mu_j=\Pro_{X_j}[\ \cdot\ |\tau \geq j]\Pro[\tau \geq j]= \Pro_{X_j}[\ \cdot\ |\tau > j-1]\Pro[\tau > j-1]$ 
						and $\nu_j=\Pro[\tau=j]\Pro_{X_j}[\ \cdot\ |\tau=j]$ for every $j<t$.
						Again we can set $\tau_t(x_t) :=\frac{d\nu_t}{d\mu_t}$ so that $\nu_t(B)=\int_B \tau_t(x)d\mu_t(x)$ for every Borel set $B\subseteq [0,1]$. We have for any Borel set $B$, 
						\begin{align}\label{Induction Step 1}
							\nu_t(B)=\Pro(\tau=t,X_t\in B)=\Pro[\tau=t]\Pro_{X_t}[X_t\in B |\tau=t].
						\end{align}
						The recursive relations between $\{\mu_t\}_{t\geq 1}$ and $\{\nu_t\}_{t\geq 1}$, coupled with Eq.\ \eqref{Induction Step 1} and the induction assumption yield
						\begin{align}\label{Induction Step 2}
							\mu_{t+1} & = \sigma_t \circ (\mu_t - \nu_t)  \nonumber \\
							& = \sigma_t \circ \left( \mu_t - \Pro_{X_t} [\ \cdot\ |\tau = t]\, \Pro[\tau =t] \right) \nonumber \\
							& = \sigma_t \circ \left( \Pro_{X_t} [\ \cdot\ |\tau \geq  t]\, \Pro[\tau \geq t] - \Pro_{X_t} [\ \cdot\ |\tau =  t]\, \Pro[\tau = t]   \right) \\
							& = \sigma_t \circ \Pro_{X_t} [\ \cdot \ | \tau > t]\, \Pro[\tau > t] \nonumber \\
							& = \Pro_{X_{t+1}}[\ \cdot \ | \tau\geq t+1] \, \Pro[\tau \geq t+1]. \nonumber
						\end{align}
						The combination of Eqs.\ \eqref{Induction Step 1} and \eqref{Induction Step 2} completes the induction step. 
						
					\end{proof}
					\seclem*
					\begin{proof}[Proof of Lemma \ref{lemma:blacwell preservingK}]
						The condition of the lemma assrests that for every $0\leq y'\leq y\leq y''\leq 1$ for the measure $\mu=\alpha\delta_{y'}+(1-\alpha)\delta_{y''}$ with expectation $y$ it holds that $\sigma\circ\delta_y\preceq_B\sigma\circ\mu$.
						
						Assume next that the measure $\nu\preceq_B \mu$. Then there exists a probability kernel $\rho:[0,1]\to\Delta([0,1])$ such that $\rho\circ\nu=\mu$. 
						Assume first that $\rho(x)$ has a support $2$ for any $x$. That is, $\rho(x)=\alpha_x\delta_{x'}+(1-\alpha_{x})\delta_{x''}$
						for $0\leq x'\leq x\leq x''\leq 1$ such that $\alpha_{x} x'+(1-\alpha_{x})x''=x.$
						Thus we can write for any Borel measurable set $B\subseteq [0,1]$
						$$\mu(B)=\int_{[0,1]}\alpha_{x}\delta_{x'}(B)+(1-\alpha_{x})\delta_{x''}(B)d\nu(x).$$ 
						
						Therefore, in particular 
						$\sigma\circ\nu=\int_{[0,1]}\sigma\circ\delta_x d\nu(x)$ and\\ 
						$\sigma\circ\mu=\int_{[0,1]}\sigma\circ(\alpha_x\delta_{x'}+(1-\alpha_{x})\delta_{x''})d\nu(x)$. In this case the result follows since $\sigma\circ\delta_x\preceq_B\sigma\circ(\alpha_x\delta_{x'}+(1-\alpha_{x})\delta_{x''})$ for every $x$.
						
						For a general $\rho$ we note that, by definition $\rho(x)\in\Delta([0,1])$ is a probability distribution on $[0,1]$ with expectation $x\in[0,1]$. The extreme points of all measures on $[0,1]$ with expectation $x$ is the set of all measures with binary support. Therefore it follows from Choquet's Theorem (see \cite{phelps2001lectures}) that $\rho(x)$ can be represented as a probability measure over binary support measures for every $x$. That is if we let $B_x\subset\Delta([0,1])$ be the set of all binary support measures with expectation $x$, then there exists $\lambda_x\in\Delta(B_x)$ such that 
						$$\rho(x)=\int_{B_x} \alpha_x\delta_{x'}+(1-\alpha_x)\delta_{x''}d\lambda_x(x',x'')$$
						
						That is, we can write 
						$\mu=\int_{[0,1]}\Big(\int_{B_x} \alpha_x\delta_{x'}+(1-\alpha_x)\delta_{x''}d\lambda_x(x',x'',\alpha_x)\Big)d\nu(x).$ Therefore since
						$\sigma\circ\nu=\int_{[0,1]}\sigma\circ\delta_x d\nu(x)$
						We can write 
						$$\sigma\circ\mu=\int_{[0,1]}\sigma\circ\rho(x)= \int_{[0,1]}\Big(\int_{B_x}\sigma\circ( \alpha_x\delta_{x'}+(1-\alpha_x)\delta_{x''})d\lambda_x(x',x'',\alpha_x)\Big)d\nu(x).$$
						The result follows since for every $x$ and every $\alpha_x\delta_{x'}+(1-\alpha_x)\delta_{x''}\in B_x$ it holds that $\delta_x\preceq_B\sigma\circ( \alpha_x\delta_{x'}+(1-\alpha_x)\delta_{x''})$.
					\end{proof}
					\thlem*
					
					\begin{proof}
						We will show that if $\sigma$ is a probability kernel that represents a conditionally independent signal, then $\sigma$ is a Blackwell preserving kernel.
						Consider first the case where $\mu = \delta_x$ for some $x \in (0,1)$ and $\nu$ is some mean-preserving spread of $\mu$. In random variable form, there exists a random variable $Y$ such that $\Ex (Y\,|\, x) = x$, where $x$ is the constant random variable supported on $\{x\}$. By \cite{AM95}, there exists a probability kernel $F_{\nu}:\{0,1\} \to [0,1]$ (referred to as a \emph{state dependent lottery} in Aumann and Maschler) so that $Y = p_x (m) :=  \Pro_x (\omega=1\,|\,m)$, where $m \in [0,1]$ is the signal whose distribution is generated by $x$ and $F_{\nu}$. Consider the random variables $X' = \Pro_x (\omega=1\,|\, s)$ and $Y' = \Pro_x (\omega=1\,|\, m, s) $, where $s$ is the signal generated by $x$ and the probability kernel $G$ corresponding to $\s$. By definition, $X' \sim \sigma \circ \mu $. Also, as the signals $m$ and $s$ may be chosen to be conditionally independent given the state $\omega \in \{0,1\}$ we have for any Borel set $B \subset [0,1]$:
						\begin{align*}
							\Pro_x [Y' \in B] & = \int_{[0,1]} \Pro_{p_x (m)} \left[ \Pro_{p_x (m)}[\omega=1\,|\, s] \in B \right] \mathrm{d}\nu (p_x (m))\\
							& = \sigma \circ \nu  (B),
						\end{align*}
						so that $Y' \sim \sigma \circ \nu$. As the tower property for conditional expectations implies $\Ex [Y'|X'] = X'$ we deduce that $\sigma \circ \nu$ is a mean-preserving spread of $ \sigma \circ \delta_x$, as required.
						
						We move on to the general case. Assume that $\mu \preceq_B \nu$ and let $\tau:[0,1] \to \Delta([0,1])$ be a probability kernel with $\nu = \tau \circ \mu$. Then, for every non-decreasing and concave $f:[0,1]\to \R$ it holds:
						\begin{align}\label{Eq. Conditional Independent Signals}
							(\sigma \circ \mu)(f) & = \int_{[0,1]}\left( \int_{[0,1]} f(t)\mathrm{d}(\sigma\circ \delta_{x})(t) \right) \mathrm{d}\mu (x) \nonumber \\
							& \geq  \int_{[0,1]}\left( \int_{[0,1]} f(t)\mathrm{d}(\sigma\circ \tau(x))(t) \right) \mathrm{d}\mu (x)\\
							& = (\sigma \circ \nu)(f) \nonumber,
						\end{align}
						where the first and last equality follow from disintegration formulas, whereas the inequality follows from the fact that $\delta_x\preceq_B \tau(x)$ together with the first part of the proof which implies $\sigma\circ \delta_x\preceq_B \sigma\circ \tau(x)$ for every $x$. As Eq.\ \eqref{Eq. Conditional Independent Signals} is equivalent to $\sigma \circ \mu \preceq_B \sigma \circ \nu$ the proof is concluded. 
					\end{proof}
					\folem*
					\begin{proof}
						Consider the case where $y'=z_{j}$, $y=z_i$, and $y''=z_{l}$ for some $j<i<l$ and $i,j,l\in Z$. As $j \leq i-1$ and $l \geq i+1$ we may write:
						\begin{align*}
							\delta_{z_{i-1}} \preceq_B \frac{z_l-z_{i-1}}{z_{l}-z_{j}}\delta_{z_{j}} + \frac{z_{i-1}-z_{j}}{z_{l}-z_{j}}\delta_{z_{l}}
						\end{align*}
						and
						\begin{align*}
							\delta_{z_{i+1}} \preceq_B \frac{z_l-z_{i+1}}{z_{l}-z_{j}}\delta_{z_{j}} + \frac{z_{i+1}-z_{j}}{z_{l}-z_{j}}\delta_{z_{l}}.
						\end{align*} 
						The above relations, coupled with the definition of $\sigma$ and simple algebraic manipulations suffice to deduce:
						\begin{align*}
							\sigma \circ \delta_{z_i} \preceq_B \frac{z_{l}-z_i}{z_{l}-z_{j}}\delta_{z_{j}}+\frac{z_i-z_{j}}{z_{l}-z_{j}}\delta_{z_{l}}
						\end{align*}
						As $\mu \preceq_B \sigma \circ \mu$ for every probability measure $\mu$, we deduce from the above relation and the fact that $\preceq_B$ is transitive that 
						\begin{align*}
							\sigma \circ \delta_{z_i} \preceq_B \sigma \circ \Bigg\lbrace \frac{z_{l}-z_i}{z_{l}-z_{j}}\delta_{z_{j}}+\frac{z_i-z_{j}}{z_{l}-z_{j}}\delta_{z_{l}} \Bigg\rbrace
						\end{align*}
						thus showing that $\s$ preserves Blackwell's order on binary-supported measures, which by Lemma \ref{lemma:blacwell preservingK} is sufficient to deduce the $\s$ is Blackwell preserving. 
					\end{proof}
					\lemexist*
					\begin{proof}
						
						Let $p_0$ be the quantile that defines the greedy policy. As mentioned above, if we let $F$ be the CDF of $\mu_1$ and $U$ is uniformly distributed random variable on $[0,1]$, then for $X=F^{-1}(U)$ it holds that $X\sim\mu_1$. It holds that $\Ex[X|U\geq p_0]=l$. Consider the function $f(p,x)=\Ex[X|U\in(p,x)]$. Note that over the domain $0\leq p<x\leq 1,$ the function $f( \cdot \ , \cdot )$ is continuous and strictly increasing both in $p$ and $x$. Let $p_1$ satisfy $F^{-1}(p_1)=l$. 
						It follows from the properties of $f$ that for every $p_0<p<p_1$ there exists a unique $x(p)$ such that $f(p,x(p))=l$. Moreover, $x(p)$ is a continuous function of $p$.  
						Therefore, the function $x(p)-p$ is continuous it attained the value $\alpha$ at 
						$p_0$ and has a left limit $0$ at $p_1$. By the intermediate value theorem it follows that for every $0<\beta\leq\alpha$ there exists 
						$(p,x(p))$ such that $\mathbb{E}[X|U\in(p,x(p))]=l$ and $x(p)-p=\beta$. Since  $x(p)-p$ is strictly decreasing in $p$ there exists a unique such $p$ for every $\beta$. Denote it by $p_{\beta}$. The two quantiles determined by $p_{\beta}$ and $x(p_{\beta})$ define a unique stopping rule $\tau$ (up to measure zero) with $\Pro(\tau=1)=\beta$.
						This concludes the proof of the lemma.
					\end{proof}
					\lemods*
					
					\begin{proof}
						
						Let $\tilde F$ and $F'$ be the CDFs of $\frac{1}{|\tilde\nu|}\tilde\nu$ and $\frac{1}{|\nu'|}\nu'$, respectively. To show $\tilde \nu\preceq_B \nu'$ it is necessary and sufficient to show that $\int_0^x\tilde F(y)dy \leq \int_0^x\tilde F'(y)dy$ for every $x\in[0,1]$. Assume by way of contradiction that $\int_0^x\tilde F(y)dy > \int_0^x\tilde F'(y)dy$
						for some $x\in[0,1]$. Then it must be that $x > \underline{y}$ since
						$\int_0^{\underline y}\tilde F(y)dy=0$.
						
						Next, we argue that $x < \overline y$. First, we note that $\tilde F (z) = 1$ for every $z \geq \overline{y}$. Thus, assuming to the contrary that $x \geq \overline y$ we obtain
						\begin{align*}
							\int_0^1 \tilde F (y) dy & = \int_0^x \tilde F(y)dy + \int_x^1 1dy\\
							& > \int_0^x  F'(y)dy + \int_x^1 F'(y)dy = \int_0^1 F' (y) dy.
						\end{align*}
						The above relation stands in contradiction with the fact that $\int_0^1 \tilde F (y) dy  = 1 - \overline{\tilde \nu} = 1 -\overline{\nu'} = \int_0^1 F' (y) dy$, thus establishing $x \in (\underline y, \overline y)$.
						
						We can further assume that  $\tilde F(x)\geq F'(x)$ because otherwise can just decrease the point $x$ until we reach such a point without violating the strict integral inequality. Let $F$ be the CDF of the probability measure $\mu_1$. On the one hand, since $\nu'\leq \mu_1$ it holds that $F'(x')- F'(x'')\leq  F(x')- F(x'')$ for every
						$x',x''\in (\underline y,\overline y)$ such that $x'>x''$. On the other hand, by construction $F(x')- F(x'')=\tilde F(x')-\tilde F(x'')$ for every $x',x''\in (\underline y,\overline y)$ such that $x'>x''$. Therefore for every $x'\in (\underline y,\overline y)$ such that $x'>x$ it holds that $$F'(x')=F'(x)+ F'(x')-F'(x)\leq \tilde F(x)+ F(x')-F(x)=\tilde F(x)+ \tilde F(x')-\tilde F(x)=\tilde F(x').$$
						Thus $F'(x')\leq \tilde F(x')$ for every $x'\in (\underline y,\overline y)$ such that $x'>x$. Therefore we must have that $\int_0^{\overline y}\tilde F(y)dy> \int_0^{\overline y}\tilde F'(y)dy$. Since $\tilde F(x)=1$ for every $x\geq\overline{y}$ we infer that $\int_0^1\tilde F(y)dy> \int_0^1\tilde F'(y)dy $ arriving at a contradiction.
						This concludes the proof of the lemma.
					\end{proof}
					
					\section{Proof of Proposition \ref{pro:sub-opt}}\label{ap:proof_pro_1}
					\firstprop*
					
					\begin{proof}
						
						Assume by way of contradiction that the greedy policy is optimal for the $\epsilon$-grid (for sufficiently small $\epsilon$, which will be specified later). For every grid point $g\in \Gamma$ and $g<l$ we denote by $v(g)$ the value of the sender (for the entire process) that initiates at $g$ under the greedy policy.
						
						We denote $c=v(l-\epsilon)<1$. We argue that 
						\begin{align}\label{eq:v}
							v(l-2\epsilon)=c^2\pm O(\epsilon^2).
						\end{align}
						To see it, we notice that if the random walk that starts at $l-2\epsilon$ reaches $l-\epsilon$ we can refer to it as if the game terminates and she receives a utility of $c$. We know that jumping to the next point above $l-\epsilon \to l$ takes an expected discount time $c$. So, jumping to the next point above $l-2\epsilon \to l-\epsilon$ should also take an expected discount time of approximately $c$. Indeed we can couple the two random walks: The one that starts at $l-\epsilon$ and reaches $l$ and the one that starts at $l-2\epsilon$ and reaches $l-\epsilon$. Failing of this coupling happens with probability $O(\epsilon^2)$; this failure occurs only if the realization of the random walk goes $l-\epsilon \to \epsilon \to l$ without visiting $l$ or $0$ in the middle. The coupled process will fail to do the same trajectory because it will move $l-\epsilon \to 0$ and $0$ is an observing state. 
						
						Using similar arguments one can show that $v(l-3\epsilon)=c^3+O(\epsilon^2)$.
						
						By the definition of value, we know that 
						\begin{align}\label{eq:v2}
							v(l-\epsilon)=\frac{\delta}{2}(1+v(l-2\epsilon))
						\end{align}
						because with probability $\frac{1}{2}$ we will reach $l$ tomorrow and the discounted value would be $\delta$ and with probability $\frac{1}{2}$ we will reach $l-2\epsilon$ tomorrow and the discounted value would be $\delta v(l-2\epsilon)$. By combining equations \eqref{eq:v} with \eqref{eq:v2} we get
						$c=\frac{\delta}{2}(1+c^2)\pm O(\epsilon^2)$ which implies that 
						\begin{align}\label{eq:c}
							c=\frac{1-\sqrt{1-\delta^2}}{\delta}\pm O(\epsilon^2).
						\end{align}
						
						Consider the initial belief $X_0$ which is distributed as follows: With probability $p=1-\frac{\epsilon}{2-2l +\epsilon}$ the belief is $l-2\epsilon$ and with probability $1-p=\frac{\epsilon}{2-2l +\epsilon}$ the belief is $1$. 
						
						If the sender uses the greedy policy at time $t=1$, she pools together $\frac{p}{4}$ of the mass located at $l-2\epsilon$ with the $1-p$ mass located at $1$ and her utility is $$1-\frac{3p}{4}+\frac{3p}{4} c^2 \pm O(\epsilon^2),$$
						where the term $\frac{3p}{4}c^2$ captures the $\frac{3p}{4}$ mass that remains at $l-2\epsilon$. 
						
						If, instead the sender stays mute at time $t=1$ and for time $t\geq 2$ she behaves greedily the following will happen. At time $t=2$ she pools together the $\frac{p}{2}$ mass at $l-\epsilon$ with the $1-p$ mass at $1$. This leaves a mass of $\frac{p}{2}$ at $l-3\epsilon$ for future utilization of adoption. In total, her value is $$\delta(1-\frac{p}{2}) + \delta \frac{p}{2} c^3 \pm O(\epsilon^2).$$
						
						In order for such a deviation from the greedy policy to be profitable we should have
						$$
						\delta(1-\frac{p}{2}) + \delta \frac{p}{2} c^3 > 1-\frac{3p}{4}+\frac{3p}{4} c^2  \pm O(\epsilon^2)
						$$
						Since $p=1-O(\epsilon)$ it is sufficient to have 
						$$
						\delta \frac{1}{2} + \delta \frac{1}{2} c^3 > \frac{1}{4}+\frac{3}{4} c^2  \pm O(\epsilon)
						$$
						
						Using Equation \eqref{eq:c} and neglecting the $O(\epsilon^2)$ and $O(\epsilon)$ error terms the inequality above becomes an inequality of $\delta$ only. One can verify that for $\delta>\frac{1}{\sqrt{2}}$ we have $\delta \frac{1}{2} + \delta \frac{1}{2} c^3 > \frac{1}{4}+\frac{3}{4} c^2$. Finally, we set $\epsilon'$ such that the total sum of all the $O(\epsilon^2)$ and $O(\epsilon)$ error terms along the proof will not exceed the gap $\delta \frac{1}{2} + \delta \frac{1}{2} c^3 - \frac{1}{4}-\frac{3}{4} c^2>0$ for every $\epsilon\leq \epsilon'$. In such cases the deviation from the greedy policy is profitable. 
					\end{proof}

					\section{Existence of a Maximum for Interval Policies}\label{sec:max-proof}
					\begin{proposition}\label{prop:existence}
						Every martingale $\bX=(X_t)_{t=1,\ldots,T}$ for the sender has an optimal interval policy. 
					\end{proposition}
					\begin{proof}
						We consider the case where $T=\infty$.
						Led $\mu$ be a positive measure on $[0,1]$ with $\mu([0,1])=r$ and let $D=\{(x,y)\in[0,r]^2:x\leq y\}$. Define a mapping $T_\mu$ from $D$ to the set of positive measures over $[0,1]$ by letting $T_\mu(x,y)$ be the measure $\nu$ that contains all mass in $\mu$ that lies between the quantiles $x$ and $y$ of $\nu$. That is for any measurable function $f$:
						$$\int_{[0,1]}f(z) d\nu(z)=\int_x^yf(F^{-1}(z))dz,$$
						where $F^{-1}$ is the inverse of the CDF of $\mu$. Note that, as in the proof of Lemma \ref{lem:exist}, it follows that for any interval measure $\nu\leq \mu$ it holds that  $\nu=T_\mu(\underline q,\overline{q})$ 
						for some $0\leq\underline q\leq\overline{q}\leq r$.
						
						It readily follows that 
						$T_\mu$ is a continuous mapping from $D$ to the set of positive measures that are endowed with the total variation norm. 
						To see this note that if $|x-x'|\leq \epsilon$, $|y-y'|\leq \epsilon$, and $f:[0,1]\to [-1,1]$, then $\|T_\mu(x,y)-T_\mu(x',y')\|_{TV}\leq 2\epsilon.$
						
						As in Lemma \ref{lem:exist} we can identify an interval policy $\{(\mu_t,\nu_t)\}_{t=1,\ldots}$ with a sequence of quantiles $\{\underline q_t,\overline{q}_t\}_{t=1,\ldots}$ such that the interval measure 
						$\nu_t\leq \mu_t$ equals $T_{\mu_t}(\underline q_t,\overline{q}_t).$
						
						Let $\{(\mu_{t,n},\nu_{t,n})_{t=1,\ldots}\}_{n=1,\ldots,\infty}$ be a sequence of interval policies that attains the supremum in the limit across all interval policies. Let $\{\{\underline q_{t,n},\overline{q}_{t,n}\}_{t=1,\ldots}\}_{n=1,\ldots}$ be the corresponding quantile representation. 
						We assume first that for every $t$ it holds that $\lim_{n\to\infty}\underline q_{t,n}=\underline{p}_t$ and $\lim_{n\to\infty}\overline{q}_{t,n}=\overline{q}_t$. Let $\{(\mu_{t},\nu_t)\}_{t=1,\ldots}$ be the measure representation of the limit policy. We claim that the policy  $\{\underline p_t,\overline{p}_t\}_{t=1,\ldots}$ achieves the optimal payoff.
						
						To see this we prove by induction 
						on $t$ that $\lim_{n\to\infty}\mu_{t,n}=\mu_t$, that $\lim_{n\to\infty}\nu_{t,n}=\nu_t$, and that $\overline{\nu_t}=l$. 
						
						Note that $\mu_{1,n}=\mu_1$ is the same measure for any $n$.
						The fact that $\lim_{n\to\infty}\nu_{1,n}=\nu_1$ follows from the above observation since $\lim_{n\to\infty}\underline q_{1,n}=\underline{p}_1$ and $\lim_{n\to\infty}\overline{q}_{1,n}=\overline{q}_1$. Since 
						$\lim_{n\to\infty}\nu_{1,n}=\nu_1$ it must hold that either $
						\nu_1$ is the zero measure or else $
						\overline{\nu_1}=l$.
						
						Assume that the claim holds for $t-1$. That is $\lim_{n\to\infty}\mu_{t-1,n}=\mu_{t-1}$, $\lim_{n\to\infty}\nu_{t-1,n}=\nu_{t-1}$, and $\overline{\nu}_{t-1}=l$. Since $\mu_{t,n}=\sigma_{t-1}\circ(\mu_{t-1,n}-\nu_{t-1,n})$ it follows that
						$\lim_{n\to\infty}\mu_{t,n}=\lim_{n\to\infty}\sigma_{t-1}\circ(\mu_{t-1,n}-\nu_{t-1,n})=\sigma_{t-1}\circ(\mu_{t-1}-\nu_{t-1})=\mu_{t}$. It therefore follows from the fact that $\lim_{n\to\infty}\underline q_{t,n}=\underline{p}_t$ and $\lim_{n\to\infty}\overline{q}_{t,n}=\overline{q}_t$ that 
						\[
						\lim_{n \to \infty} \nu_{n,t} = \lim_{n \to \infty} T_{\mu_{n,t}}(\underline{q}_{t,n}, \overline{q}_{t,n}) = T_{\mu_{t}}(\underline{q}_{t}, \overline{q}_{t}) = \nu_{t}.
						\] In addition since $\overline{\nu}_{t,n}=l$ we have that either $\nu_t=0$ or $\overline{\nu}_t=l$ as desired.
						Since the utility of the sender from each policy is 
						$$\sum_{t=1}^\infty w_t|\nu_{t,n}|$$
						we have that the utility converges to the utility of the limit policy $\sum_{t=1}^\infty w_t|\nu_{t}|$.
						The claim that the limit policy achieves the optimum now readily follows.
						
						We next show that the quantile converges assumption holds without loss. We can take the original sequence of policies $\{(\mu_{t,n},\nu_{t,n})_{t=1,\ldots}\}_{n=1,\ldots}$
						and take a subsequence $\{(\mu_{t,n_{i_1}},\nu_{t,n_{i_1}})_{t=1,\ldots}\}_{i_1=1,\ldots}$ such that  $\lim_{i_1\to\infty}\underline q_{1,n_{i_1}}=\underline{p}_1$ and $\lim_{i_1\to\infty}\overline{q}_{1,n_{i_1}}=\overline{q}_1$.
						
						We proceed inductively to get a sequence of refinements such that\\ $\{(\mu_{t,n_{i_k}},\nu_{t,n_{i_k}})_{t=1,\ldots}\}_{i_k=1,\ldots}$ is a refinement of $\{(\mu_{t,n_{i_{k-1}}},\nu_{t,n_{i_{k-1}}})_{t=1,\ldots}\}_{i_2=1,\ldots}$ and\\ 
						$\lim_{i_2\to\infty}\underline q_{k,n_{i_k}}=\underline{p}_k$ and $\lim_{i_k\to\infty}\overline{q}_{k,n_{i_2}}=\overline{q}_k$. It is now easy to see that the diagonal subsequence has the desired properties. 
					\end{proof}
					
				\end{document}